\documentclass[]{article}
\usepackage{amsmath, amssymb, dsfont, mathrsfs,amsthm}
\usepackage{float}
\usepackage{mdframed}
\usepackage{appendix}
\usepackage[hidelinks]{hyperref}
\usepackage{graphicx}
\usepackage{caption}
\usepackage[margin=1in]{geometry}
\allowdisplaybreaks

\title{Coded Trace Reconstruction
}

\author{Mahdi Cheraghchi\thanks{Department of Computing, Imperial College London, UK. Email: m.cheraghchi@imperial.ac.uk} \and Ryan Gabrys\thanks{Department of Electrical and Computer Engineering, University of Illinois, Urbana-Champaign, USA. Email: ryan.gabrys@gmail.com}\and Olgica Milenkovic\thanks{Department of Electrical and Computer Engineering, University of Illinois, Urbana-Champaign, USA. Email: milenkov@illinois.edu}\and Jo\~ao Ribeiro\thanks{
Department of Computing, Imperial College London, UK. Email: j.lourenco-ribeiro17@imperial.ac.uk}
\thanks{Research supported by the DARPA Molecular Informatics Program, the NSF and SRC SemiSynBio Award, and NSF grant 1618366. This work was supported in part by the Center for Science of Information (CSoI), an NSF Science and Technology Center, under grant agreement CCF-0939370.}
}

\date{}

\newcommand{\cC}{\mathcal{C}}

\newcommand{\cS}{\mathcal{S}}
\newcommand{\cT}{\mathcal{T}}

\newcommand{\E}{\mathds{E}}

\newcommand{\eps}{\epsilon}

\newtheorem{thm}{Theorem}
\newtheorem{prop}[thm]{Property}

\newtheorem{lem}[thm]{Lemma}

\newtheorem{defn}[thm]{Definition}

\newcommand{\poly}{\textnormal{poly}}
\newcommand{\bits}{\{0,1\}}
\newcommand{\Egood}{E_{\textsf{good}}}
\newcommand{\Enc}{\mathsf{Enc}}
\newcommand{\Dec}{\mathsf{Dec}}

\let\originalleft\left
\let\originalright\right
\renewcommand{\left}{\mathopen{}\mathclose\bgroup\originalleft}
\renewcommand{\right}{\aftergroup\egroup\originalright}

\usepackage{url}
\begin{document}

\maketitle

\begin{abstract}
Motivated by average-case trace reconstruction and coding for portable DNA-based storage systems, we initiate the study of \emph{coded trace reconstruction}, the design and analysis of high-rate efficiently encodable codes that can be efficiently decoded with high probability from few reads (also called \emph{traces}) corrupted by edit errors. Codes used in current portable DNA-based storage systems with nanopore sequencers are largely based on heuristics, and have no provable robustness or performance guarantees even for an error model with i.i.d.\ deletions and constant deletion probability. 
Our work is a first step towards the design of efficient codes with provable guarantees for such systems. 
We consider a constant rate of i.i.d.\ deletions, and perform an analysis of marker-based code-constructions. 
This gives rise to codes with redundancy $O(n/\log n)$ (resp.\ $O(n/\log\log n)$) that can be efficiently reconstructed from $\exp(O(\log^{2/3}n))$ (resp.\ $\exp(O(\log\log n)^{2/3})$) traces, where $n$ is the message length.
Then, we give a construction of a code with $O(\log n)$ bits of redundancy that can be efficiently reconstructed from $\poly(n)$ traces if the deletion probability is small enough.
Finally, we show how to combine both approaches, giving rise to an efficient code with $O(n/\log n)$ bits of redundancy which can be reconstructed from $\poly(\log n)$ traces for a small constant deletion probability.
\end{abstract}

\section{Introduction}\label{sec:intro}

Trace reconstruction was originally introduced by Batu, Kannan, Khanna, and McGregor~\cite{BKKM04}, motivated by problems in sequence alignment, phylogeny, and computational biology. The setting for the problem is as follows: There is an unknown string $x\in\{0,1\}^n$, and our goal is to reconstruct it. Towards this goal, we are allowed to ask for \emph{traces} of $x$, which are obtained by sending $x$ through a deletion channel. This channel independently deletes bits of $x$ with a given deletion probability $d$. As a result, each trace corresponds to a subsequence of $x$. We wish to minimize the number of traces required for reconstructing $x$ with high probability.

Since its introduction, the problem of trace reconstruction has been studied from several different perspectives. Two of the main perspectives correspond to \emph{worst-case} trace reconstruction~\cite{BKKM04, HMPW08, MPV14, DOS17, NP17}, where the reconstruction algorithm must work simultaneously for all strings in $\{0,1\}^n$, and \emph{average-case} trace reconstruction~\cite{BKKM04, HMPW08, MPV14, HHP17, PZ17, HPP18}, where the reconstruction algorithm is only required to work with high probability, taken over the choice of string and the randomness of the reconstruction algorithm for a uniformly random string. The number of traces required for average-case trace reconstruction is, as expected, much smaller than that required for worst-case trace reconstruction. The problem in question has also been studied from a combinatorial coding perspective~\cite{Lev01,SGSD17,GY18,HY19}.

The above described results on average-case trace reconstruction can be interpreted from a coding-theoretic perspective: They state that there exist very large codebooks which can be reconstructed efficiently from relatively few traces. However, no efficient encoders are known for such codes, and it may be possible to further reduce the number of traces required for reconstruction by relaxing the size of the code.

This point of view naturally leads to the problem of \emph{coded trace reconstruction}: The goal is to design high rate, efficiently encodable codes whose codewords can be efficiently reconstructed with high probability from very few traces with constant deletion probability. Here, ``high rate'' refers to a rate approaching $1$ as the block length increases. We remark that in such a case, the number of traces must grow with the block length of the code. Coded trace reconstruction is also closely related to and motivated by the read process in portable DNA-based data storage systems, which we discuss below.

\paragraph{Motivation}
A practical motivation for coded trace reconstruction comes from portable DNA-based data storage systems using DNA nanopores, first introduced in~\cite{YGM17}. 
In DNA-based storage, a block of user-defined data is first encoded over the nucleotide alphabet $\{A,C,G,T\}$, and then transformed into moderately long strands of DNA through a DNA synthesis process. For ease of synthesis, the DNA strands are usually encoded to have balanced $GC$-content, so that the fraction of $\{{A,T\}}$ and $\{{G,C\}}$ bases is roughly the same. To recover the block of data, the associated strand of DNA is sequenced with nanopores, resulting in multiple corrupted reads of its encoding. Although the errors encountered during nanopore sequencing include both deletions/insertions as well as substitution errors, careful read preprocessing alignment~\cite{YGM17} allows the processed reads to be viewed as traces of the data block's encoding. 
As a result, recovering the data block in question can be cast in the setting of trace reconstruction. Due to sequencing delay constraints\footnote{Current Oxford nanopore sequencers can process roughly 450 nucleotides per second.}, it is of great interest to minimize the number of reads required to reconstruct the data block.

The trace reconstruction procedures associated to the codes used by practical portable DNA-based storage systems~\cite{YGM17, OAC+18} are largely based on heuristics. The trace reconstruction algorithm proposed in~\cite{YGM17} operates on carefully designed coded strings, 
but makes use of multiple sequence alignment algorithms which are notoriously difficult to analyze rigorously. The trace reconstruction algorithm from~\cite{OAC+18} does not make use of specific read-error correction codes and is a variation of the Bitwise Majority Alignment (BMA) algorithm originally introduced in~\cite{BKKM04}. However, the BMA algorithm is only known to be robust when the errors correspond to i.i.d.\ deletions and the fraction of errors is at most $O(1/\log n)$, where $n$ denotes the blocklength of the code. 
Moreover, the proposed codes have been designed only for a fixed blocklength. As a result, the codes from~\cite{YGM17,OAC+18} have no robustness or performance guarantees for trace reconstruction even under i.i.d.\ deletions with constant deletion probability. Consequently, our work on coded trace reconstruction is a first step towards the development of codes with provable robustness guarantees and good performance for trace reconstruction in portable DNA-based data storage systems.

\subsection{Related work}

Recently, there has been significant interest both in trace reconstruction and coding for settings connected to DNA-based data storage.

Regarding trace reconstruction, a chain of works~\cite{BKKM04,HMPW08,MPV14,HHP17,PZ17,HPP18} has succeeded in substantially reducing the number of traces required for average-case trace reconstruction. The state-of-the-art result, proved by Holden, Pemantle, and Peres~\cite{HPP18}, states that $\exp(\log^{1/3} n)$ traces suffice to reconstruct a random $n$-bit string under arbitrary constant deletion probability. Much less is known about the worst-case trace reconstruction problem. The current best upper bound of $\exp(O(n^{1/3}))$ traces was proved concurrently by De, O'Donnell, and Servedio~\cite{DOS17} and Nazarov and Peres~\cite{NP17} through algorithms that rely only on single-bit statistics from the trace. They also showed that this upper bound is tight for this restricted type of algorithms. The gap between upper and lower bounds for trace reconstruction for both the average- and worst-case settings is still almost exponential. The state-of-the-art lower bounds, obtained by Chase~\cite{Cha19} and improving on the previous best lower bounds by Holden and Lyons~\cite{HL18}, are close to $\log^{5/2} n$ traces for average-case trace reconstruction and $n^{3/2}$ traces in the worst-case setting. Trace reconstruction has also been studied over a large class of sticky channels~\cite{MDSG16}, motivated by nanopore sequencers. A channel is said to be \emph{sticky} if it preserves the block structure of the input, i.e., no input runs are completely deleted, and no new runs are added to it. In particular, the deletion channel is not a sticky channel. 

Another line of work has focused on combinatorial settings related to coded trace reconstruction~\cite{Lev01,SGSD17,GY18,HY19}. The reported works study the number of traces required for exact reconstruction when each trace is subjected to a bounded number of adversarial edit errors, and the string may be assumed to belong to a code from some general class. We note that Levenshtein~\cite{Lev01} also studies a version of probabilistic trace reconstruction where the string is sent through a memoryless channel, but the deletion channel is not included in this family of channels. Other closely related combinatorial reconstruction problems consist in recovering a string from a subset of its substrings~\cite{kiah2016codes,GM17,GM18}.

Coded trace reconstruction is also related to the multi-use deletion channel. In fact, the decoding problem for this channel can be interpreted as coded trace reconstruction with a fixed number of traces. Some results are known about the capacity of this channel for small deletion probability~\cite{HM14}, and about maximum likelihood decoding~\cite{SDDF18}.

Some recent works have focused on coding for other channel models inspired by DNA-based storage. 
The model studied in~\cite{LSWY18,SC18,SRB18,LSWY19} views codewords as sets comprised of several sequences over some alphabet, and the (adversarial) errors consist of erasure and insertion of sequences in the set, as well as deletions, insertions, or substitutions within some of the sequences.
Each sequence in the set corresponds to a different DNA strand whose contents are reconstructed via high throughput sequencing-based procedures that introduce errors. The goal is to recover the whole set from the erroneous sequences. This is fundamentally different from our model, as we focus on the correct reconstruction of a long single strand of DNA from multiple sequencing attempts and high probabilistic deletion error rates, which is especially relevant for portable nanopore DNA-based data storage systems.
Another model similar to the one described above, motivated by the permutation channel, is studied in~\cite{KT18}. 
However, codewords there consist of multisets of symbols in some alphabet, and errors are comprised of deletions, insertions, substitutions, and erasures of symbols in the multiset. 
An information-theoretic treatment of related but abstracted models of DNA-based data storage may be found in~\cite{HSRT17,SH19}. Very recently, a model for clustering sequencing outputs according to the relevant DNA strand and codes that allow for correct clustering have been studied in~\cite{SYLW19}.

Although we focus on portable DNA-based storage systems with nanopore sequencers~\cite{YGM17,OAC+18}, there has also been significant activity on practical aspects of other types of DNA-based storage, e.g.,~\cite{GBC+13,GHP+15,YYM+15}. General overviews of the field can be found in~\cite{YKG+15,MGKY18}.

To conclude this section, we note that some novel results on trace reconstruction have been developed concurrently to this work. 
Davies, Racz, and Rashtchian~\cite{DRR19} study trace reconstruction over trees (in the graph-theoretic sense), where standard trace reconstruction corresponds to recovering a path. 
Ban et al.\ \cite{BCFSS19} consider the problem of trace reconstruction of distributions over small sets of strings (a setting commonly known as population recovery). 
Krishnamurthy et al.\ \cite{KMMP19} explore alternative and generalized forms of trace reconstruction (e.g., trace reconstruction over matrices and sparse strings). 
As mentioned before, Chase~\cite{Cha19} improved on the trace reconstruction lower bounds by Holden and Lyons~\cite{HL18}. Finally, the problem of coded trace reconstruction has also been studied concurrently and independently in a similar setting by Abroshan et al.\ \cite{AVDG19}. 
They consider a code obtained by concatenating several blocks, each block being a codeword of a Varshamov-Tenengolts code correcting one deletion. 

\subsection{Channel model}

The channel model used for representing nanopore systems may be summarized as follows. For a given input string $x\in\bits^n$, a deletion probability $d$, and an integer $t(n)$, the channel returns $t(n)$ \emph{traces} of $x$. Each trace of $x$ is obtained by sending $x$ through a deletion channel with deletion probability $d$, i.e., the deletion channel deletes each bit of $x$ independently with probability $d$, and outputs a subsequence of $x$ containing all bits of $x$ that were not deleted in order. The $t(n)$ traces are independent and identically distributed as outputs of the deletion channel for input $x$. 

Given a code $\cC\subseteq\bits^n$, we say that $\cC$ can be \emph{efficiently reconstructed from $t(n)$ traces} if there exists a polynomial $p(n)$ and a polynomial-time algorithm that recovers every $c\in\cC$ from $t(n)$ traces of $c$ with probability at least $1-1/p(n)$ over the probability distribution of the traces.

\subsection{Our contributions}

We initiate the study of coded trace reconstruction for efficient, high-rate codes against a constant rate of deletions. More specifically:
\begin{itemize}
	\item We analyze the performance of marker-based constructions coupled with worst-case trace reconstruction algorithms. These constructions have the advantage that they can be easily adapted to work with a large range of inner codes. 
	
At a high level, the construction operates by splitting an $n$-bit message into short blocks of length $O(\log^2 n)$, 
encoding each block with an inner code satisfying a certain constraint, and adding markers of length $O(\log n)$ between the blocks. 
The structure of the markers and the property of the inner code imply that, with high probability, we can split the traces into many shorter sub-traces associated with substrings of length $O(\log^2 n)$, and then apply the worst-case trace reconstruction algorithm on the sub-traces. Our main result in this context is Theorem~\ref{thm:tracesmarker1}.
	\begin{thm}\label{thm:tracesmarker1}
		For every constant deletion probability $d<1$, there exists an efficiently encodable code $\cC\subseteq\bits^{n+r}$ with redundancy $r=O(n/\log n)$ that can be efficiently reconstructed from $\exp(O(\log^{2/3}n))$ traces.
	\end{thm}
	We significantly improve on this simple construction for small constant deletion probabilities with a more involved approach described in the second part of the paper. The above construction is relevant as it shows that we can instantiate the marker-based construction with any inner code satisfying a simple constraint and iterate the marker-based construction by further dividing each block of length $\log^2 n$ into blocks of length $(\log\log n)^2$ and adding markers of length $O(\log\log n)$ between them. In this setting, it is almost guaranteed that reconstruction of a small fraction of blocks will fail. Nevertheless, this problem can be easily resolved by adding error-correction redundancy to the string to be encoded. This leads to the following result, which can be extended beyond two marker levels.
	\begin{thm}\label{thm:leveltwo}
		For every constant deletion probability $d<1$, there exists an efficiently encodable code $\cC\subseteq\bits^{n+r}$ with redundancy $r=O(n/\log\log n)$ that can be efficiently reconstructed from $\exp(O(\log\log n)^{2/3})$ traces.
	\end{thm}
	
	\item We take advantage of the fact that we can instantiate the marker-based constructions with a large range of inner codes to construct high-rate marker-based codes over the $\{A,C,G,T\}$ alphabet with two important properties: The codes have balanced $GC$-content and provably require few traces to be efficiently reconstructed. We follow the same ideas as the marker-based constructions above, but with different markers and an inner code over a larger alphabet and with stronger constraints. In this context, we obtain the following results.
	\begin{thm}\label{thm:quatbasic}
		For every constant deletion probability $d<1$, there exists an efficiently encodable code $\cC\subseteq\{A,C,G,T\}^{n+r}$ with redundancy $r=O(n/\log n)$ and balanced $GC$-content that can be efficiently reconstructed from $\exp(O(\log^{2/3}n))$ traces.
	\end{thm}
	\begin{thm}\label{thm:quatleveltwo}
		For every constant deletion probability $d<1$, there exists an efficiently encodable code $\cC\subseteq\{A,C,G,T\}^{n+r}$ with redundancy $r=O(n/\log\log n)$ and balanced $GC$-content that can be efficiently reconstructed from $\exp(O(\log\log n)^{2/3})$ traces.
	\end{thm}
	
	\item The result of Theorem~\ref{thm:tracesmarker1} may be further improved by considering a more careful design of the high-rate inner code to be used in the marker-based constructions, provided that the deletion probability is a small enough constant. This allows for using a modified version of an algorithm for average-case trace reconstruction described in~\cite{HMPW08} which leads to a substantial reduction in the number of traces required for reconstruction and barely any rate changes. As a first step towards achieving this goal, we first design a low-redundancy code that can be efficiently reconstructed from polynomially many traces. The proposed coding scheme relies on the fact that we can efficiently encode $n$-bit messages into strings that are almost \emph{subsequence-unique} (see Definition~\ref{def:subseq}) via explicit constructions of \emph{almost $k$-wise independent spaces} (see Section~\ref{sec:almostind}). The average-case trace reconstruction algorithm from~\cite{HMPW08} operates on subsequence-unique strings, and a simple adaptation of the algorithm suffices for our approach. In summary, we have the following result.
	\begin{thm}\label{thm:polytraces}
		If the deletion probability is a small enough constant, there exists an efficiently encodable code $\cC\subseteq\bits^{n+r}$ with redundancy $r=O(\log n)$ that can be efficiently reconstructed from $\poly(n)$ traces.
	\end{thm}
	An important step in our analysis is to show how to adapt this code for use as an inner code in the marker-based construction. Some care is needed, since the global structure of the strings we deal with changes significantly due to the presence of the markers. In particular, the bootstrapping method in the trace reconstruction algorithm from~\cite{HMPW08} no longer works, and we must find a way to circumvent this issue. Our findings for this scenario are described in the next theorem.
	\begin{thm}\label{thm:redtrace}
		If the deletion probability is a small enough constant, there exists an efficiently encodable code $\cC\subseteq\bits^{n+r}$ with redundancy $r=O(n/\log n)$ that can be efficiently reconstructed from $\poly(\log n)$ traces.
	\end{thm}
\end{itemize}

For simplicity, our exposition mostly focuses on constructions of \emph{binary} codes, although it provides some guidelines and simple coding procedures for quaternary codes. One should note that coded trace reconstruction is inherently harder for smaller alphabets.

\subsection{Organization}

The paper is organized as follows: In Section~\ref{sec:prelims}, we define relevant notation and discuss known results that we find useful in our subsequent derivations. We describe and analyze general marker-based constructions in Section~\ref{sec:marker}. Then, we show how to reduce the number of traces required for a small deletion probability in Section~\ref{sec:reducetrace}. We discuss some open problems in Section~\ref{sec:open}.

\section{Notation and preliminaries}\label{sec:prelims}

\subsection{Notation}

We denote the length of a string $x$ by $|x|$, and its Hamming weight by $w(x)=|\{i:x_i\neq 0\}|$. Given two strings $x$ and $y$ over the same alphabet, we denote their concatenation by $x||y$. For a string $x$, we define $x[a,b)=(x_a,x_{a+1},\dots,x_{b-1})$ and $x[a,b]=(x_a,x_{a+1},\dots,x_b)$. If $|x|=n$, we define $x[a,\cdot]=(x_a,x_{a+1},\dots,x_n)$. We say that $y$ is a subsequence of $x$ if there exist indices $i_1<i_2<\cdots<i_{|y|}$ such $x_{i_j}=y_j$. Moreover, $y$ is said to be a substring of $x$ if $y=x[a,a+|y|)$ for some $1\leq a\leq |x|-|y|+1$. Given two strings $x,y\in\bits^n$, we write $x+y$ for the bitwise XOR of $x$ and $y$. A \emph{run of length $\ell$} in a string $x$ is a substring of $x$ comprising $\ell$ identical symbols. Sets are denoted by calligraphic letters such as $\cS,\cT$. Random variables are denoted by uppercase letters such as $X$, $Y$, and $Z$. The uniform distribution over $\bits^t$ is denoted by $U_t$, and the binomial distribution on $n$ trials with success probability $p$ is denoted by $\mathsf{Bin}(n,p)$. The binary entropy function is denoted by $h$ and all logarithms $\log$ are taken with respect to the base $2$.

\subsection{Almost $k$-wise independent spaces}\label{sec:almostind}

We start by defining almost $k$-wise independence and present a related result that we will find useful in our future derivations.

\begin{defn}[$\eps$-almost $k$-wise independent random variable]
	A random variable $X\in\{0,1\}^m$ is said to be \emph{$\eps$-almost $k$-wise independent} if for all sets of $k$ distinct indices $i_1, i_2,\dots,i_k$ we have
	\begin{equation*}
		|\Pr[X_{i_i}=x_1,\dots,X_{i_k}=x_k]-2^{-k}|\leq \eps
	\end{equation*}
	for all $(x_1,\dots,x_k)\in\{0,1\}^k$.
\end{defn}

The following result gives an efficient construction of an $\eps$-almost $k$-wise independent space which can be generated from few uniformly random bits. 
\begin{lem}[\cite{AGHP92}]\label{lem:almostind}
	For every $m$, $k$, and $\eps$ there exists an efficiently computable function $g:\{0,1\}^t\to\{0,1\}^m$ with $t=O\left(\log\left(\frac{k\log m}{\eps}\right)\right)$ such that $g(U_t)$ is an $\eps$-almost $k$-wise independent random variable over $\{0,1\}^m$, where $U_t$ denotes the uniform distribution over $\{0,1\}^t$. 
\end{lem}

\subsection{Nearly-optimal systematic codes for edit errors}

We require systematic codes that are robust against edit errors (deletions and insertions). Nearly-optimal systematic codes for adversarial edit errors have been recently constructed using optimal protocols for deterministic document exchange~\cite{Hae18,CJLW18}. The following result is relevant to our analysis.
\begin{lem}[\cite{Hae18,CJLW18}]\label{lem:systcode}
	For every $m$ and $t<m$ there exists an efficiently encodable and decodable systematic code $\cC_{\mathsf{edit}}\subseteq\{0,1\}^{m+r}$ with encoder $\Enc_{\mathsf{edit}}:\bits^m\to\bits^{m+r}$ and redundancy $r=O\left(t\log^2\frac{m}{t}+t\right)$ that can efficiently correct up to $t$ edit errors. In particular, if $t=\Theta(m)$ then the redundancy is $r=O(m)$.
\end{lem}

\subsection{Trace reconstruction}

Next, we discuss several results pertaining to the worst-case and average-case trace reconstruction problem that will be useful for our constructions.

\subsubsection{Worst-case trace reconstruction}

For worst-case reconstruction, the state-of-the-art result used in Section~\ref{sec:marker} is summarized below.
\begin{lem}[\cite{DOS17,NP17}]\label{lem:worstcase}
	For every $n$ and constant deletion probability $d$ there exists an algorithm that reconstructs an arbitrary string $x\in\{0,1\}^n$ with probability at least $1-\exp(-2n)$ from $\exp(O(n^{1/3}))$ traces in time $\exp(O(n^{1/3}))$.
\end{lem}

\subsubsection{Trace reconstruction of subsequence-unique strings}\label{sec:tracerecuniquesubseq}

One of the key tools for our constructions in Section~\ref{sec:reducetrace} is a modified version of the efficient trace reconstruction algorithm~\cite{HMPW08} for what we refer to as \emph{subsequence-unique} strings. This algorithm may also be used for average-case trace reconstruction. We start by defining subsequence-unique strings.
\begin{defn}[$w$-subsequence-unique string]\label{def:subseq}
	A string $x\in\bits^n$ is said to be \emph{$w$-subsequence-unique} if for every $a$ and $b$ such that either $a<b$ or $b+1.1w<a+w$ we have that the substring $x[a,a+w)$ is not a subsequence of $x[b,b+1.1w)$.
\end{defn}
Note that these strings have been under the name ``substring-unique" in~\cite{HMPW08}. We proposed the name change to avoid confusion with a different definition under the same name, described in~\cite{GM18}. The following result about subsequence-unique strings was established in~\cite{HMPW08}.
\begin{lem}[\protect{\cite[Theorem 2.5]{HMPW08}}]\label{lem:tracerecsubsequnique}
	For $w=100\log n$ and a small enough constant deletion probability $d$, there exists an algorithm that reconstructs every $w$-subsequence-unique string $x\in\bits^n$ with probability $1-1/\poly(n)$ from $\poly(n)$ traces in time $\poly(n)$.
\end{lem}
Since a uniformly random string is $w$-subsequence-unique with high probability, Lemma~\ref{lem:tracerecsubsequnique} applies to average-case trace reconstruction. As we make explicit use of the algorithm behind Lemma~\ref{lem:tracerecsubsequnique}, for the sake of clarity we provide next a more in-depth discussion of the method. However, before we proceed to the actual description of the algorithm, we briefly introduce some definitions and basic related results. 

Given integers $i$ and $j$ and a deletion probability $d$, we denote the probability that the $i$-th bit of a string appears as the $j$-th bit of its trace by $P(i,j)$. Then, we have
\begin{equation*}
	P(i,j)=\binom{i-1}{j-1}(1-d)^jd^{i-j}.
\end{equation*}
The following lemma states some useful properties of $P(i,j)$.
\begin{lem}[\protect{\cite[Lemma 2.2]{HMPW08}}]\label{lem:basicpij}
	If $j\leq (1-3d)i$, then $P(i,j)\geq 2\sum_{i'>i} P(i',j)$. Furthermore, if $(1-4d)i<j<(1-3d)i$, we have $P(i,j)\geq \exp(-6di)$.
\end{lem}
Intuitively, the second part of Lemma~\ref{lem:basicpij} means that we have a good idea of the position of $x_i$ in the trace if $i$ is small. The following result makes use of this. It states that we can recover the first $O(\log n)$ bits of an arbitrary string with $\poly(n)$ traces, which is required to bootstrap the trace reconstruction algorithm from~\cite{HMPW08}.
\begin{lem}[\protect{\cite[Theorem 2.1]{HMPW08}]}]\label{lem:recoverfirst}
	Fix a string $x\in\bits^n$, and suppose that we know $x_1,\dots,x_{h-1}$. Then, there is an algorithm that recovers $x_h$ from $\exp(O(hd\log(1/d)))$ traces of $x$ with probability $1-o(1)$, provided that $d<1/3$.
\end{lem}
In the second part of the algorithm, we must look for matchings of certain strings within the traces. To this end, we introduce the following definition.
\begin{defn}[Matching]
	Fix a string $x\in\bits^n$, and let $T$ denote its trace. Then, we say that there is a \emph{matching of $x[a,b)$ in $T$} if there exists some $u$ such that $T[u-(b-a),u)=x[a,b)$.
\end{defn}
Matchings of $w$-subsequence-unique strings have useful properties, as formalized in the following lemma.
\begin{lem}[\protect{\cite[Lemma 2.4]{HMPW08}}]\label{lem:matchsubsequnique}
	If $x$ is $w$-subsequence-unique and there is a matching of $x[a,a+w)$ in $T$, say at $T[u-w,u)$, then the probability that $T_{u-1}$ does not come from $x[a+w,a+1.1w)$ is at most $nd^{0.001 w}$. 
\end{lem}
We are now in a position to describe the algorithm introduced in~\cite{HMPW08}. We begin by setting $w=100\log n$, $v=w/d$, and $j=(v-0.1w)(1-3d)$. Then, to recover a $w$-subsequence-unique string $x$ we proceed with two steps: First, we use the algorithm from Lemma~\ref{lem:recoverfirst} to recover the first $v$ bits of $x$ with $\poly(n)$ traces. 
Now, suppose we have recovered $x_1,\dots,x_{i-1}$ for $i-1\geq v$. Our next goal is to recover $x_i$ with $\poly(n)$ traces. Note that if $i$ is relatively large, we cannot use the algorithm from Lemma~\ref{lem:recoverfirst} to recover $x_i$ anymore, as it would require more than $\poly(n)$ traces. To achieve our goal, we instead focus on finding matchings of the substring $x[i-v-w,i-v)$ within the trace. Let $T$ denote a trace of $x$, and suppose there is a matching of $x[i-v-w,i-v)$ in $T$ at positions $T[u-w,u)$. Then, we set $V=T[u,\cdot]$, i.e., we let $V$ be the suffix of the trace following the matching. The key property is that $\Pr[V_j=1]$ satisfies a threshold property depending on the value of $x_i$. More precisely, there exist two positive values $B_1>B_0$ sufficiently far apart such that $\Pr[V_j=1]\leq B_0$ if $x_i=0$ and $\Pr[V_j=1]\geq B_1$ if $x_i=1$. Moreover, all terms in these inequalities can be estimated with a small error from $\poly(n)$ traces of $x$. As a result, we can reliably estimate $x_i$ by checking whether $\Pr[V_j=1]\leq B_0$ or $\Pr[V_j=1]\geq B_1$.

We prove next the threshold property for $\Pr[V_j=1]$. Let $R$ denote the position in $x$ of the bit appearing in position $u-1$ in the trace $T$ of the matching for $x[i-v-w,i-v)$. In other words, $R$ denotes the position in $x$ of the last bit appearing in the matching in $T$. We may write
\begin{align*}
	\Pr[V_j=1]&=\sum_{r=1}^n \Pr[R=r]\Pr[V_j=1|R=r]\\
	&=\eps_i(x)+\sum_{r=i-v}^{i-v+0.1w} \Pr[R=r]\Pr[V_j=1|R=r]\\
	&=\eps_i(x)+\sum_{r=i-v}^{i-v+0.1w} \Pr[R=r]\sum_{\ell=r+1}^n P(\ell-r,j)x_\ell\\
	&=\eps_i(x)+\sum_{r=i-v}^{i-v+0.1w} \Pr[R=r]\sum_{\ell=r+1}^{i-1} P(\ell-r,j)x_\ell\\
	&+\sum_{r=i-v}^{i-v+0.1w} \Pr[R=r]\left(P(i-r,j)x_i+\sum_{\ell=i+1}^{n} P(\ell-r,j)x_\ell\right),
\end{align*}
where the second equality follows from Lemma~\ref{lem:matchsubsequnique} with $0\leq \eps_i(x)\leq nd^{-0.001w}$. Using the first part of Lemma~\ref{lem:basicpij}, we conclude that $\sum_{\ell=i+1}^{n}P(\ell-r,j)\leq \frac{1}{2}P(i-r,j)$. As a result, we have
\begin{multline}\label{eq:UBoriginal}
x_i=0\implies \Pr[V_j=1]\leq \eps_i(x)+\sum_{r=i-v}^{i-v+0.1w}\Pr[R=r]\sum_{\ell=r+1}^{i-1} P(\ell-r,j)x_\ell
+\frac{1}{2}\sum_{r=i-v}^{i-v+0.1w}\Pr[R=r]P(i-r,j)
\end{multline}
and
\begin{multline}\label{eq:LBoriginal}
x_i=1\implies \Pr[V_j=1]\geq \eps_i(x)+\sum_{r=i-v}^{i-v+0.1w}\Pr[R=r]\sum_{\ell=r+1}^{i-1} P(\ell-r,j)x_\ell
+\sum_{r=i-v}^{i-v+0.1w}\Pr[R=r]P(i-r,j).
\end{multline}
By the second part of Lemma~\ref{lem:basicpij}, since $i-r\leq v$ and $v=w/d$, we have $P(i-r,j)\geq 2^{-9w}$. Combining this with Lemma~\ref{lem:matchsubsequnique} for $d$ small enough means that the gap between the right hand side of~\eqref{eq:UBoriginal} and~\eqref{eq:LBoriginal} is at least $2^{-(9w+1)}$. To finalize the argument, we note that (i) we can efficiently approximate $\Pr[V_j=1]$ to within an error of, say, $2^{-100w}$ with high probability from $\poly(n)$ traces of $x$, and (ii) we can efficiently approximate $\Pr[R=r|R<i]$ to within the same error given that we know $x_1,\dots,x_{i-1}$, provided $d$ is small enough. Since $\Pr[R<i]\geq 1-nd^{-0.001w}$ by Lemma~\ref{lem:matchsubsequnique}, we can further efficiently approximate $\Pr[R=r]$ to within an error of, say, $2^{-50w}$ with high probability. From these observations it follows that we can estimate $x_i$ correctly with high probability from $\poly(n)$ traces, where the degree of the polynomial is independent of $i$, as desired.

\section{Marker-based constructions}\label{sec:marker}

We start with simple constructions of high-rate codes that can be efficiently reconstructed from a few traces. The idea behind the approach is the following: Each codeword contains markers, consisting of sufficiently long runs of $0$'s and $1$'s. Between two consecutive markers, we add a short block containing a codeword from an inner code satisfying a mild constraint. 

Intuitively, the runs in the markers will still be long in the trace, and so we hope to be able to correctly identify the positions of all markers in a trace with high probability. After this is done, we can effectively split the trace into many shorter, independent sub-traces corresponding to a block (and possibly some bits from the two markers delimiting it). Then, we can apply worst-case trace reconstruction algorithms to the sub-traces. The savings in the number of traces required for reconstruction stem from the fact that sub-traces are short, and that each trace can be utilized simultaneously (and independently) by all blocks. This idea for reconstruction almost works as is, except that the process of identifying the markers in a trace may be affected by long runs of $0$'s originating from a block between two markers. However, this can be easily solved by requiring that all runs of $0$'s in each block are short enough. Many codes, including codes with low redundancy, satisfy the desired property, and hence make for good candidates for the inner code.

We describe and analyze a code based on the idea discussed above in Section~\ref{sec:basicconst}. 
Then, we consider a follow-up construction in Section~\ref{sec:leveltwo} which requires fewer traces, at the expense of a decrease in the rate. 
At a high-level, this second code is obtained by introducing two levels of markers and adding some simple error-correction redundancy to the message prior to other encodings.
Finally, in Section~\ref{sec:synthesisconstraints} we extend these ideas to the $\{A,C,G,T\}$ alphabet in order to obtain high-rate codes with desireable properties for use in DNA-based storage. Namely, these codes have balanced $GC$-content and can be reconstructed from few traces. Such codes are designed by exploiting the fact that the marker-based constructions can be instantiated with a large range of inner codes, and we can make the inner code satisfy stronger constraints than before.

\subsection{A simple construction}\label{sec:basicconst}
Here we provide a precise description of the encoder $\Enc$ for our code $\cC$ and prove Theorem~\ref{thm:tracesmarker1}. For simplicity, we consider $d=1/2$ throughout. 

Let $\ell=50\log n$, and define two strings $M_0=0^\ell$ and $M_1=1^\ell$. Then, a marker $M$ is a string of length $2\ell$ of the form $M=M_0||M_1=0^\ell||1^\ell$. We also require an efficiently encodable and decodable inner code $\cC'\subseteq\bits^{m+r}$ with encoder $\Enc':\bits^m\to\bits^{m+r}$, where $m=\log^2 n$ and $r$ is the redundancy, satisfying the following property.
\begin{prop}\label{prop:manyones}
	For all $c\in\cC$ and substrings $s$ of $c$ with $|s|=\sqrt{m}$, it holds that $w(s)\geq |s|/3$.
\end{prop}
In other words, every codeword of $\cC'$ has many $1$'s in all short enough substrings. Such efficient codes exist with redundancy $r=O(\log m)=O(\log\log n)$, which is enough for our needs. We provide a simple construction in Section~\ref{sec:inner code}.

Suppose we wish to encode an $n$-bit message $x\in\{0,1\}^n$. The encoder $\Enc$ on input $x$ proceeds through the following steps:
\begin{enumerate}
	\item Split $x$ into $n/\log^2 n$ blocks, each of length $\log^2 n$
	\begin{equation*}
	x=x^{(1)}||x^{(2)}||\cdots||x^{(n/\log^2 n)};
	\end{equation*}
	
	\item Encode each block $x^{(i)}$ under the inner code $\cC'$ to obtain $\overline{x}^{(i)}=\Enc'(x^{(i)})\in\{0,1\}^{\log^2 n+r}$;
	
	\item Set the encoding of $x$, denoted by $\Enc(x)$, to be
	\begin{equation*}
		\Enc(x)=1^\ell||\overline{x}^{(1)}||M||\overline{x}^{(2)}||M||\cdots||M||\overline{x}^{(n/\log^2 n)}||0^\ell.
	\end{equation*}
\end{enumerate}
We remark that the first run $1^\ell$ and the last run $0^\ell$ are superfluous, and are added only to make the analysis simpler. Computing $\Enc(x)$ from $x$ and decoding $x$ from $\Enc(x)$ can both be done efficiently if the inner code $\cC'$ is efficiently encodable and decodable. 

We now compute the redundancy of $\cC$. It is straightforward to see that
\begin{equation}\label{eq:redmarker}
	|\Enc(x)|\leq \frac{n}{\log^2 n}(|M|+|\overline{x}^{(1)}|)=n+O\left(\frac{n}{\log n}\right)+\frac{nr}{\log^2 n}.
\end{equation}
As mentioned before, we have $r=O(\log\log n)$. Therefore, $\cC$ can be made to have redundancy $O\left(\frac{n}{\log n}\right)$. In the remainder of this section, we prove Theorem~\ref{thm:tracesmarker1} using $\cC$ via a sequence of lemmas. For convenience, we restate the theorem below.
\begin{thm}[Theorem~\ref{thm:tracesmarker1}, restated]
	There is an efficient algorithm that recovers every $c\in\cC$ from $\exp(O(\log^{2/3} n))$ traces in time $\poly(n)$ with probability $1-1/\poly(n)$.
\end{thm}
To prove Theorem~\ref{thm:tracesmarker1} we proceed in steps: First, we show that the markers $M$ still contain long enough runs after they are sent through the deletion channel. 
Then, we show that no long runs of $0$'s originate from the sub-traces associated with each block. This implies that we can correctly identify the position of the ``01" string of each marker 
in the trace. Finally, we show that we can apply the worst-case trace reconstruction algorithm from Lemma~\ref{lem:worstcase} to recover each block with high probability and with the desired number of traces.

We start by proving that the markers $M$ still contain long runs after they are sent through the deletion channel.
\begin{lem}\label{lem:largemarkers}
	Let $0^{L_0}1^{L_1}$ be the output of the deletion channel on input $M$. Then,
	\[
	\Pr[L_0>10\log n,L_1>0]\geq 1-n^{-3}.
	\]
\end{lem}
\begin{proof}
	The result follows by a standard application of the Chernoff bound. More precisely, we have $\mathds{E}[L_0]=25\log n$, and hence
	\[
	\Pr[L_0\leq 10\log n]=\Pr[L_0\leq \mathds{E}[L_0]-15\log n]\leq \exp\left(-\frac{15^2 \log^2 n}{2\mathds{E}[L_0]}\right)\geq n^{-4}.
	\]
	To conclude the proof, we note that $\Pr[L_1=0]=2^{-\ell}=n^{-50}$, and that the two events in question are independent.
\end{proof}
We now show that no long runs of $0$'s originate from the sub-traces associated with each block.
\begin{lem}\label{lem:nolongruns}
	Let $c\in\cC'$. Then, a trace of $c$ does not contain a run of $0$'s of length at least $10\log n$ with probability at least $1-n^{-3}$.
\end{lem}
\begin{proof}
	Since $c\in\cC'$, a run of $0$'s of length at least $10\log n$ in the trace of $c$ requires that at least $10\times \frac{\log n}{3}-1$ consecutive $1$'s are deleted in $c$. The probability that this happens for a fixed sequence of $10\times \frac{\log n}{3}-1$ consecutive $1$'s is at most $n^{-3.3}$. Since there are at most $O(\log^2 n)$ such sequences in $c$, by the union bound it follows that the desired probability is at most $n^{-3}$.
\end{proof}
The next lemma follows immediately by combining Lemmas~\ref{lem:largemarkers} and~\ref{lem:nolongruns} with the union bound over the $n/\log^2 n$ blocks.
\begin{lem}\label{lem:splittrace}
	Consider the following event $E$: \emph{We correctly identify the separation between the traces of $0^\ell$ and $1^\ell$ from every marker in the trace of $\Enc(x)$ by looking for all $1$'s that appear immediately after a run of at least $10\log n$ $0$'s.}
	
	Then, $E$ happens with probability at least $1-n^{-2}$ over the randomness of the trace.
\end{lem}

We are now ready to prove Theorem~\ref{thm:tracesmarker1}. Let $E$ denote the event described in Lemma~\ref{lem:splittrace}. Then, Lemma~\ref{lem:splittrace} implies that, conditioned on $E$ happening, we can split a trace $T$ of $\Enc(x)$ into $n/\log^2 n$ strings $T^{(1)},\dots,T^{(n/\log^2 n)}$ satisfying the following:
	\begin{itemize}
		\item The strings $T^{(i)}$ are independent;
		\item Each string $T^{(i)}$ is distributed like a trace of $1^\ell || \overline{x}^{(i)} ||0^\ell$ conditioned on the high probability event $E$.
	\end{itemize}
	In fact, each string $T^{(i)}$ can be identified by looking for the $(i-1)$-th and $i$-th runs of $0$ of length at least $10\log n$ in the trace $T$, and picking every bit in $T$ immediately after the $(i-1)$-th run up to and including the $i$-th run.
	
	Observe that $1^\ell || \overline{x}^{(i)} ||0^\ell$ has length $O(\log^2 n)$. Suppose that we have $t=\exp(O(\log n)^{2/3})$ independent traces $T_1,\dots,T_t$ of $\Enc(x)$. Let $E_\mathsf{all}$ denote the event that $E$ holds for all $T_i$ simultaneously. Combining Lemma~\ref{lem:splittrace} with a union bound yields
	\begin{equation}\label{eq:boundEall}
		\Pr[E_\mathsf{all}]\geq 1- t/n^2 >1-1/n.
	\end{equation}
Fix some trace reconstruction algorithm $\mathcal{A}$, and let $E^{(i)}_{\mathsf{indFail}}$ denote the event that $\mathcal{A}$ fails to recover a fixed string $y^{(i)}=1^\ell||\overline{x}^{(i)}||0^\ell$ from $t$ independent traces of $y^{(i)}$. 
Assuming that $E_\mathsf{all}$ holds, the strings $T^{(i)}_1,\dots,T^{(i)}_t$ are distributed as $t$ independent traces of $y^{(i)}$, each also satisfying the conditions that the first run $1^\ell$ is not completely deleted, the last run $0^\ell$ has length at least $10\log n$ in the trace, and there is no run of $0$'s of length at least $10\log n$ in the trace of $\overline{x}^{(i)}$. 
We denote the event that these conditions hold for all of the $t$ independent traces of $y^{(i)}$ by $E^{(i)}_\mathsf{split}$. 
Finally, we let $E_\mathsf{fail}$ denote the event that we fail to recover $\Enc(x)$ from the $t$ i.i.d.\ traces $T_1,\dots, T_t$. Then, we have
	\begin{align}
		\Pr[E_\mathsf{fail}]&\leq \Pr[E_\mathsf{fail},E_\mathsf{all}]+\Pr[\neg E_\mathsf{all}]\nonumber\\
		& =\Pr[(\exists i: E^{(i)}_\mathsf{indFail}), (\forall i: E^{(i)}_\mathsf{split})]+\Pr[\neg E_\mathsf{all}]\nonumber\\
		&\leq \Pr[\exists i: E^{(i)}_\mathsf{indFail}]+1/n\label{eq:genbeforeunion}\\
		&\leq \sum_{i=1}^{n/\log^2 n}\Pr[E^{(i)}_\mathsf{indFail}]+1/n.\label{eq:genboundfail}
	\end{align}
	The first equality follows from the discussion in the previous paragraph, the second inequality follows from~\eqref{eq:boundEall}, and the third inequality follows 
	by the union bound. Instantiating $\mathcal{A}$ with the worst-case trace reconstruction algorithm from Lemma~\ref{lem:worstcase}, we conclude from~\eqref{eq:genboundfail} that
	\begin{equation*}
		\Pr[E_\mathsf{fail}]\leq n\cdot \exp(-2\log^2 n)+1/n<2/n.
	\end{equation*}

	As a result, we can successfully recover $x$ from $\exp(O(\log n)^{2/3})$ traces of $\Enc(x)$ with probability at least $1-2/n$. To conclude the proof, 
	we note that we can repeat the process $O(\log n)$ times and take the majority vote to boost the success probability to $1-1/p(n)$ for any fixed polynomial $p$ of our choice. The total number of traces required is still $\exp(O(\log^{2/3}n))$. 
	Since recovering each $\overline{x}^{(i)}$ from the associated traces takes time $\exp(O(\log^{2/3}n))$ and the inner code $\cC'$ has an efficient decoder, the whole procedure is efficient.

\subsubsection{Instantiating the inner code}\label{sec:inner code}

What remains to be done is to instantiate the inner code $\cC'$ with the appropriate parameters and properties. To this end, we present a simple construction of an efficiently encodable and decodable inner code $\cC'$ with encoder $\Enc':\bits^m\to\bits^{m+r}$ and redundancy $r=O(\log m)$. We can then obtain the desired code by setting $m=\log^2 n$. The starting point is the following result.
\begin{lem}\label{lem:manyones}
	Let $g:\bits^t\to\bits^m$ be the function whose existence is guaranteed by Lemma~\ref{lem:almostind} with $k=3w$ and $\eps=2^{-10w}$ for $w=100\log m$ (hence $t=O(\log m)$). Fix some $x\in\bits^m$ and consider the random variable $Y=x+g(U_t)$. Then, with probability at least $1-2/m$, we have that $Y$ satisfies the following property:
	\begin{prop}\label{prop:manyonessmall}
		$w(Y[a,a+w))\geq 0.4w$ simultaneously for all $1\leq a\leq m-w+1$.
	\end{prop}
\end{lem}
\begin{proof}
	Fix some $a$. Then, we have
	\begin{align*}
		\Pr[w(Y[a,a+w))< 0.4w]&=\sum_{y:w(y)<0.4w} \Pr[Y[a,a+w)=y]\\
		&\leq \sum_{y:w(y)<0.4w} (2^{-w}+2^{3w}\eps)\\
		&\leq 2^{wh(0.4)}\cdot 2^{-w+1}\\
		&\leq \frac{2}{m^2}.
	\end{align*}
	The first inequality follows because $Y$ is $\eps$-almost $k$-wise independent, and the second inequality follows from a standard bound on the volume of the Hamming ball and the fact that $2^{3w}\eps<2^{-w}$. Since there are at most $m$ choices for $a$, by the union bound we conclude that $Y$ fails to satisfy the desired property with probability at most $m\cdot 2/m^2=2/m$, as desired.
\end{proof}
Given $x\in\bits^m$, we evaluate $\Enc'(x)$ as follows: We iterate over all $z\in\bits^t$ until we find $z$ such that $y=x+g(z)$ satisfies $w(s[a,a+w))\geq 0.4w$. Such a string $z$ is known to exist by Lemma~\ref{lem:manyones} and can be found in time $\poly(m)$ since $t=O(\log m)$. Then, we set $\Enc'(x)=z||x+g(z)$. 

Observe that the redundancy of $\cC'$ is exactly $|z|=t=O(\log m)$, and that we have encoders and decoders for $\cC'$ running in time $\poly(m)$ since $t=O(\log m)$. To see that $\cC'$ satisfies the property required in this section, fix some substring $s$ of $\Enc'(x)$ such that $|s|=\sqrt{m}$. Then, $w(s)\geq 0.4w\cdot |s|/w-t\geq 0.39|s|$ 
provided that $m$ is large enough.

Finally, we remark that the code used in this marker-based construction is just an example of a viable inner code $\cC'$. Any \emph{structured family} of codes satisfying Property~\ref{prop:manyones} may be used instead, and one may envision adding more constraints to $\cC'$, depending on the application constraints at hand. We exploit this fact in Sections~\ref{sec:leveltwo} and~\ref{sec:synthesisconstraints}. For example, in Section~\ref{sec:synthesisconstraints} we will require that $\cC'$ is a code over $\{A,C,G,T\}$ satisfying an analogue of Property~\ref{prop:manyones} while also having balanced $GC$-content.

\subsection{Adding a second level of markers}\label{sec:leveltwo}

In our next construction, we exploit the fact that the marker-based construction from Section~\ref{sec:basicconst} can be instantianted with a large range of inner codes to prove Theorem~\ref{thm:leveltwo}. To do so, we show that we can iterate the marker-based construction so that we can split a trace into even smaller sub-traces with high probability. This leads to a code requiring fewer traces, but with a penalty in the redundancy. We restate Theorem~\ref{thm:leveltwo} for convenience.
\begin{thm}[Theorem~\ref{thm:leveltwo}, restated]
	There exists an efficiently encodable code $\cC_0\subseteq\bits^{n_0+r_0}$ with encoder $\Enc_0:\bits^{n_0}\to\bits^{n_0+r_0}$ and redundancy $r_0=O(n_0/\log\log n_0)$ that can be efficiently reconstructed from $\exp(O(\log\log n_0)^{2/3})$ traces with probability at least $1-2/n_0$.
\end{thm}

As before, for simplicity we set $d=1/2$ throughout the section. We will use the same construction blueprint as in Section~\ref{sec:basicconst}, except for the following differences:
\begin{itemize}
	\item We assume the $n$-bit message $x$ belongs to a binary code $\cC_{\mathsf{Ham}}\subseteq\bits^n$ with encoder 
	$\Enc_{\mathsf{Ham}}:\bits^{n_0}\to\bits^n$ and relative (Hamming) distance\footnote{The relative Hamming distance of a code is defined as its minimum Hamming distance normalized by its block length.} $30/\log^2 n_0$. 
	In particular, we have $x=\Enc_0(x_0)$ for some $x_0\in\bits^{n_0}$. 
	
	Such efficiently encodable and decodable codes are known to exist with redundancy $n-n_0=O\left(n_0\frac{\log\log n_0}{\log n_0}\right)$ (see Appendix~\ref{app:hamming} for a proof).
The reasons for using this encoding will be made clear later;	
	\item The inner code $\cC'$ differs from the one used in Section~\ref{sec:inner code}.
\end{itemize}
 
If $\cC$ denotes the code obtained via the reasoning of Section~\ref{sec:basicconst} and $\Enc$ corresponds to its encoder, then the encoder $\Enc_0:\bits^{n_0}\to\bits^{n_0+r_0}$ for our final code $\cC_0$ is obtained by composing the encoders of $\cC_{\mathsf{Ham}}$ and $\cC$, i.e.,
\begin{equation*}
	\Enc_0=\Enc\circ \Enc_{\mathsf{Ham}}.
\end{equation*}

We proceed to describe the encoder $\Enc'$ for the inner code $\cC'$ of $\cC$. Given $y\in\bits^m$, where $m=\log^2 n$, we split $y$ into $m/\log^2 m$ blocks of length $\log^2 m$,
\begin{equation*}
	y=y^{(1)}||y^{(2)}||\cdots||y^{(m/\log^2 m)}.
\end{equation*}
Then, we take $\cC''\subseteq\bits^{m'+r'}$ with encoder $\Enc'':\bits^{m'}\to \bits^{m'+r'}$ as the efficiently encodable and decodable code constructed in Section~\ref{sec:inner code} with message length $m'=\log^2 m$ and redundancy $r'=O(\log m')=O(\log\log m)$. For each $i$, we define $\overline{y}^{(i)}=\Enc''(y^{(i)})$. 
Moreover, we let $\ell'=50\log m$, and define the marker $M'=0^{\ell'}||1^{\ell'}$. Then, we define $\Enc'(y)$ as
\begin{equation*}
	\Enc'(y)=M'||\overline{y}^{(1)}||M'||\overline{y}^{(2)}||M'||\cdots || M' || \overline{y}^{(m/\log^2 m)}||M'.
\end{equation*}
Observe that we can efficiently decode $y$ from $\Enc'(y)$ provided that $\cC''$ is efficiently decodable.

We first compute the redundancy of the inner code $\cC'$ and the resulting code $\cC$ obtained as in Section~\ref{sec:basicconst}. We have
\begin{equation*}
	|\Enc'(y)|=m+\frac{m}{\log^2 m}\cdot \left(|M'|+O\left(\log\log m\right)\right)=m+O\left(\frac{m}{\log m}\right).
\end{equation*}
Thus, $\cC'$ has redundancy $r=O(m/\log m)$. Plugging $r$ into~\eqref{eq:redmarker} and recalling that $m=\log^2 n$, we conclude that $\cC$ has redundancy
\begin{equation*}
	O\left(\frac{n}{\log n}\right)+O\left(\frac{n\log^2 n}{\log^2 n\cdot \log\log n}\right)=O\left(\frac{n}{\log\log n}\right).
\end{equation*}
As a result, since $n=n_0+O\left(n_0\frac{\log\log n_0}{\log n_0}\right)$, the code $\cC_0$ has redundancy $r_0=O(n_0/\log\log n_0)$, as desired.

We now show that $\cC'$ satisfies Property~\ref{prop:manyones}. First, we observe that $\cC''$ satisfies Property~\ref{prop:manyonessmall} with $m'$ in place of $m$. Then, since each $M'$ has weight $0.5|M'|$, we conclude that every substring $s$ of $\Enc'(y)$ such that $|s|=\sqrt{m}$ satisfies
\begin{equation*}
	w(s)\geq 0.4w\cdot |s|/w-\ell'\geq 0.39|s|,
\end{equation*}
provided $m$ is large enough, since $\ell'=O(\log m)$. As a result, Lemma~\ref{lem:splittrace} holds for this choice of inner code, and we can hence focus solely on the trace reconstruction problem for strings of the form
\begin{equation}\label{eq:specformtwolevel}
	1^\ell || \Enc'(y)||0^\ell=1^\ell||M'||\overline{y}^{(1)}||M'||\cdots||M'||\overline{y}^{(m/\log^2 m)}||M'||0^\ell,
\end{equation}
where $\ell=O(\log n)=O(\sqrt{m})$, and provided the number of traces used is significantly smaller than $n$.

We now give a trace reconstruction algorithm for strings of the form~\eqref{eq:specformtwolevel} that requires $\exp(O(\log^{2/3}m))=\exp(O(\log\log n_0)^{2/3})$ traces and time, and succeeds with probability at least $1-1/\poly(m)=1-1/\poly(\log n_0)$. We have the following two lemmas whose proofs are analogous to those of Lemmas~\ref{lem:largemarkers} and~\ref{lem:nolongruns} and hence omitted.
\begin{lem}\label{lem:largemarkers2}
	Let $0^{L_0}1^{L_1}$ be the output of the deletion channel on input $M'$. Then,
	\[
	\Pr[L_0>10\log m,L_1>0]\geq 1-m^{-3}.
	\]
\end{lem}
\begin{lem}\label{lem:nolongruns2}
	Let $c\in\cC''$. Then, a trace of $c$ does not contain a run of $0$'s of length at least $10\log m$ with probability at least $1-m^{-3}$.
\end{lem}
Combining Lemmas~\ref{lem:largemarkers2} and~\ref{lem:nolongruns2} with the union bound leads to the following analogue of Lemma~\ref{lem:splittrace}.
\begin{lem}\label{lem:splittrace2}
	Consider the following event $E'$: \emph{We correctly identify the separation between the traces of $0^{\ell'}$ and $1^{\ell'}$ from every marker in the trace of $\Enc'(x)$ by looking for all $1$'s that appear immediately after a run of at least $10\log m$ $0$'s.}
	
	Then, $E'$ happens with probability at least $1-m^{-2}$ over the randomness of the trace.
\end{lem}

As in Section~\ref{sec:basicconst}, Lemma~\ref{lem:splittrace2} implies that, conditioned on $E'$ happening for a trace $T$ of $1^\ell || \Enc'(y)||0^\ell$, we can split $T$ into independent sub-traces $T^{(i)}$ each distributed like a trace of $1^{\ell'}||\Enc''(y^{(i)})||0^{\ell'}$ conditioned on the high probability event $E'$. 

Let $\mathcal{A}$ denote the worst-case trace reconstruction algorithm from Lemma~\ref{lem:worstcase} for strings of length $O(m')=O(\log^2 m),$ with failure probability at most $\exp(-\Omega(\log^2 m))$. A reasoning similar to that preceding~\eqref{eq:genboundfail} with Lemma~\ref{lem:splittrace2} in place of Lemma~\ref{lem:splittrace}, and the code $\cC'$ designed in this section in place of $\cC$ shows that, using algorithm $\mathcal{A}$, we fail to recover $\Enc'(y)$ from $\exp(O(\log^{2/3}m))$ i.i.d.\ traces of $1^\ell || \Enc'(y)||0^\ell$ with probability at most
\begin{equation}\label{eq:boundaprime}
	m\cdot \exp(-\Omega(\log^2 m))+1/m < 2/m.
\end{equation}

Let $\mathcal{A}'$ denote the algorithm that recovers $\Enc'(y)$ from $\exp(O(\log^{2/3}m))$ i.i.d.\ traces of $1^\ell || \Enc'(y)||0^\ell$ with failure probability at most $2/m$ as described above. We hope to instantiate~\eqref{eq:genboundfail} directly with $\mathcal{A}'$ to obtain the desired upper bound on the reconstruction failure probability for $\cC$. However, this approach does not produce a satisfactory result as the failure probability of $\mathcal{A'}$ is $2/m=1/\poly(\log n)$, which is too large to be used in the union bound.

Recall from Section~\ref{sec:basicconst} that, given $x\in\bits^n$, the codeword $\Enc(x)$ of $\cC$ is obtained by splitting $x$ into $n/\log^2 n$ blocks $x^{(i)}$ and encoding each block with the encoder $\Enc'$ associated with $\cC'$. From the discussion in the previous paragraph, a fraction of blocks $x^{(i)}$ will be reconstructed with errors. Below we argue that this fraction is of size at most $10/\log^2 n_0$ with probability at least $1-2/n_0$. The reasoning is similar in spirit to that used to derive~\eqref{eq:genbeforeunion}, and it suffices to complete the proof of Theorem~\ref{thm:leveltwo}. In fact, suppose we recovered $\tilde{x}$, which is a guess of $x$ with at most a $(10/\log^2 n_0)$-fraction of incorrect blocks. In particular, the relative Hamming distance between $\tilde{x}$ and $x$ is at most $10/\log^2 n_0$. Since the relative distance of $\cC_{\mathsf{Ham}}$ is at least $30/\log^2 n_0$ and we assumed that $x\in\cC_{\mathsf{Ham}}$, it follows that $\Dec_\mathsf{Ham}(\tilde{x})=\Dec_\mathsf{Ham}(x)=x_0$. Therefore, we conclude that we can recover the underlying message $x_0$ with probability at least $1-2/n_0$ from $\exp(O(\log^{2/3}m))=\exp(O(\log\log n_0)^{2/3})$ i.i.d.\ traces of $\Enc_0(x)$. This proves Theorem~\ref{thm:leveltwo}.

As the last step, we show that the fraction of bad blocks is small enough with high probability. Suppose that we have access to $t=\exp(O(\log^{2/3}m))$ i.i.d.\ traces $T_1,\dots,T_t$ of $\Enc(x)$, where $\Enc$ is the encoder associated with $\cC$. Let $E$ denote the event from Lemma~\ref{lem:splittrace}, and let $E_\mathsf{all}$ denote the event that $E$ holds for all $T_i$ simultaneously. 
As before, assuming that $E_\mathsf{all}$ holds, the strings $(T^{(i)}_1,\dots, T^{(i)}_t)_{1\leq i\leq n/\log^2 n}$ are independent between all $i$, and each tuple of strings $T^{(i)}_1,\dots, T^{(i)}_t$ is distributed as $t$ independent traces of $1^\ell || \Enc'(x^{(i)})||0^\ell$, each $T^{(i)}_j$ also satisfying the conditions that the first run $1^\ell$ is not completely deleted, the last run $0^\ell$ has length at least $10\log n$ in the trace, and no run of $0$'s has length at least $10\log n$ in the trace of $\Enc'(x^{(i)})$.
Denote the event that both these conditions hold for $t$ independent traces of $1^\ell || \Enc'(x^{(i)})||0^\ell$ by $E^{(i)}_\mathsf{split}$. Invoking the trace reconstruction algorithm $\mathcal{A'}$ defined above, let $I^{(i)}_\mathsf{indFail}$ denote the indicator random variable of the event that $\mathcal{A'}$ fails to recover $1^\ell || \Enc'(x^{(i)})||0^\ell$ from $t$ independent traces of $1^\ell || \Enc'(x^{(i)})||0^\ell$. Taking into account the previous discussion, we let $E_\mathsf{fail}$ denote the probability that more than a $(10/\log^2 n_0)$-fraction of blocks $x^{(i)}$ is recovered with errors. Then, we have
\begin{align}\label{eq:genfailindicator}
	\Pr[E_\mathsf{fail}]&\leq \Pr[E_\mathsf{fail},E_\mathsf{all}]+\Pr[\neg E_\mathsf{all}]\nonumber\\
	& = \Pr\left[\sum_{i=1}^{n/\log^2 n} I^{(i)}_\mathsf{indFail}>\frac{n}{\log^2 n}\cdot \frac{10}{\log^2 n_0},\forall i: E^{(i)}_\mathsf{split}\right]+\Pr[\neg E_\mathsf{all}]\nonumber\\
	&\leq \Pr\left[\sum_{i=1}^{n/\log^2 n} I^{(i)}_\mathsf{indFail}>\frac{n}{\log^2 n}\cdot \frac{10}{\log^2 n_0}\right]+1/n_0.
\end{align}
The first equality follows from the discussion in the previous paragraph, and the second inequality follows from Lemma~\ref{lem:splittrace} and the fact that $n>n_0$. Recalling~\eqref{eq:boundaprime}, which asserts that the failure probability for $\mathcal{A'}$ is at most $2/m$, shows that
\begin{equation*}
	\Pr[I^{(i)}_\mathsf{indFail}]\leq 2/m=2/\log^2 n<2/\log^2 n_0
\end{equation*}
holds for every $i$. Since the $I^{(i)}_\mathsf{indFail}$ are independent for all $i$, a standard application of the Chernoff bound yields the following lemma.
\begin{lem}\label{lem:smallfraction}
	We have
	\begin{equation*}
		\Pr\left[\sum_{i=1}^{n/\log^2 n} I^{(i)}_\mathsf{indFail}>\frac{n}{\log^2 n}\cdot \frac{10}{\log^2 n_0}\right] \leq n_0^{-10}.
	\end{equation*}
\end{lem}
We remark that the Chernoff bound yields a stronger upper bound than the one featured in Lemma~\ref{lem:smallfraction}. However, for simplicity we use a weaker upper bound that still suffices for our needs. 
Combining~\eqref{eq:genfailindicator} with Lemma~\ref{lem:smallfraction} allows us to conclude that $\Pr[E_\mathsf{fail}]<2/n_0$, as desired.

\subsection{A code for DNA-based data storage decodable from a few traces}\label{sec:synthesisconstraints}

We describe next how to adapt the ideas from Sections~\ref{sec:basicconst} and~\ref{sec:leveltwo} and combine them with techniques from~\cite{YKGM18} in order to construct codes over the alphabet $\{A,C,G,T\}$ that have balanced $GC$-content and provably require few traces for reconstruction. As already pointed out, strings with balanced $GC$-content are significantly easier to synthesize than their non-balanced counterparts. Therefore, constructions accomodating this constraint are well-suited for use in DNA-based data storage.

The constructions follow those outlined in Sections~\ref{sec:basicconst} and~\ref{sec:leveltwo}. The only modifications are the choice of markers and the definition of the inner code. We focus on discussing these changes and their properties within the setting of Section~\ref{sec:basicconst}. The full argument and the extension for the two-level marker-based construction of Section~\ref{sec:leveltwo} follow in a straightforward manner.

We first describe the modified markers. The marker $M$ used throughout the section is of the form $M=(AC)^\ell || (TG)^\ell$, where $\ell=25\log n$ and $n$ is the message length. Observe that this marker has the same length as the original marker in Section~\ref{sec:basicconst}. Moreover, $M$ has balanced $GC$-content.

In order to proceed as in Section~\ref{sec:basicconst} we need to design an efficiently encodable and decodable inner code $\cC'\subseteq\{A,C,T,G\}^{m'}$ with balanced $GC$-content which satisfies a property analogous to Property~\ref{prop:manyones}. 

Suppose that $\cC'$ has encoder $\Enc':\bits^m\to\{A,C,T,G\}^{m'}$ and that $m'=m/2+r$, where $m=\log^2 n$ as in Section~\ref{sec:basicconst} and $r$ denotes the redundancy to be determined. Given the composition of $M$, the property we wish $\cC'$ to satisfy is the following:
\begin{prop}\label{prop:quat}
	For all $c\in\cC'$ and substrings $s$ of $c$ with $|s|=\sqrt{m}$, it holds that at least $|s|/3$ symbols of $s$ are $T$ or $G$.
\end{prop}
Similarly to Lemma~\ref{lem:nolongruns}, it can be shown that if $\cC'$ satisfies Property~\ref{prop:quat}, then with high probability a trace of $c\in\cC'$ will not contain long runs consisting only of symbols $A$ and $C$. As a result, with high probability we can easily split a trace into many sub-traces associated with different blocks as in Section~\ref{sec:basicconst}. This is accomplished by looking for all long substrings of the trace consisting only of $A$'s and $C$'s in the trace. The reason is that, with high probability, each such substring consists of the trace of an $(AC)^\ell$ substring from a marker $M$ possibly with some extra symbols prepended. In that case we can correctly identify the separation between the traces of $(AC)^\ell$ and $(TG)^\ell$ in all markers by looking for the first $T$ or $G$ after every sufficiently long substring of $A$'s and $C$'s.

We proceed to describe the encoder $\Enc'$ of the inner code $\cC'$ that has redundancy $r=O(\log m)$. We combine a technique from~\cite{YKGM18} with the code from Section~\ref{sec:inner code}. As an additional ingredient in the construction, we require an efficiently encodable and decodable binary balanced code $\cC_1$ with encoder $\Enc_1:\bits^{m/2}\to\bits^{m/2+r_1}$. Nearly-optimal constructions of such codes are known, and they have redundancy $r_1=O(\log m)$~\cite{Knu86,IW10}. Let $\cC_2\subseteq \bits^{m/2+r_2}$ denote the code from Section~\ref{sec:inner code} with encoder $\Enc_1:\bits^{m/2}\to\bits^{m/2+r_2}$ and redundancy $r_2=O(\log m)$. By padding one of $\cC_1$ or $\cC_2$ appropriately, we may assume that $r_1=r_2=r$, i.e., that both codes have the same block length. Similarly to~\cite{YKGM18}, we define the bijection $\Psi:\bits^{n}\times\bits^n \to\{A,C,G,T\}^n$ as
\begin{equation*}
	\Psi(a,b)_i=\begin{cases}
	A,\textrm{ if $(a_i,b_i)=(0,0),$}\\
	T,\textrm{ if $(a_i,b_i)=(0,1),$}\\
	C,\textrm{ if $(a_i,b_i)=(1,0),$}\\
	G,\textrm{ if $(a_i,b_i)=(1,1).$}
	\end{cases}
\end{equation*}
The code $\cC'$ is defined via an encoding $\Enc':\bits^m\to\{A,C,G,T\}^{m/2+r}$ of the form
\begin{equation*}
	\Enc'(x)=\Psi(\Enc_1(x^{(1)}),\Enc_2(x^{(2)})),
\end{equation*}
where $x=x^{(1)}||x^{(2)}\in\bits^{m/2}\times\bits^{m/2}$. It is clear that decoding $x$ from $\Enc'(x)$ can be performed efficiently. We hence have the following lemma.
\begin{lem}\label{lem:quatinnercode}
	The inner code $\cC'$ has balanced $GC$-content and satisfies Property~\ref{prop:quat}.
\end{lem}
\begin{proof}
	Suppose that $c=\Psi(c_1,c_2),$ where $c_1\in \cC_1$ and $c_2\in\cC_2$. To see that $c$ has balanced $GC$-content, note that the number of $C$'s and $G$'s in $c$ equals the weight of $c_1$. We have $w(c_1)=|c_1|/2$ since $\cC_1$ is a balanced code, and hence $c$ has balanced $GC$-content. To verify that $\cC$ satisfies Property~\ref{prop:quat}, note that the number of $T$'s and $G$'s within a substring $c[i,j]$ equals $w(c_2[i,j])$. Since $\cC_2$ satisfies Property~\ref{prop:manyones}, the proof follows.
\end{proof}

Given Lemma~\ref{lem:quatinnercode}, we can now proceed along the steps described in Section~\ref{sec:basicconst} by splitting a trace of $\cC$ into many short sub-traces associated with different blocks, and then applying a worst-case trace reconstruction algorithm on each block. We remark that although the algorithm from Lemma~\ref{lem:worstcase} works for worst-case trace reconstruction over binary strings, it can be easily adapted for quaternary strings. In fact, if $t$ traces suffice for a worst-case trace reconstruction algorithm to reconstruct a string in $\bits^n$ with high probability, then a simple modification of this procedure recovers any quaternary string in $\{A,C,G,T\}^n$ with $2t$ traces. This is achieved by mapping the symbols in the first $t$ traces over $\{A,C,G,T\}$ to traces over $\bits$ according to, say, $A\mapsto 0, C\mapsto 0, G\mapsto 1, T\mapsto 1$, and the symbols in the last $t$ traces according to $A\mapsto 0, C\mapsto 1, G\mapsto 0, T\mapsto 1$. We can now run the binary worst-case algorithm on both sets of $t$ traces, and recover the original string over $\{A,C,G,T\}$ from the two outputs.

Taking into account the previous discussion, applying the reasoning from Section~\ref{sec:basicconst} to the marker $M$ and inner code $\cC'$ defined in this section leads to Theorem~\ref{thm:quatbasic}, which we restate for completeness.
\begin{thm}[Theorem~\ref{thm:quatbasic}, restated]
	For every deletion probability $d<1$, there exists an efficiently encodable code $\cC\subseteq\{A,C,G,T\}^{n+r}$ with redundancy $r=O(n/\log n)$ and balanced $GC$-content that can be efficiently reconstructed from $\exp(O(\log n)^{2/3})$ traces.
\end{thm}

Following the reasoning from Section~\ref{sec:leveltwo} with the modified markers and $\cC''$ instantiated with the inner code $\cC'$ we designed in this section proves Theorem~\ref{thm:quatleveltwo}, which we also restate for completeness.
\begin{thm}[Theorem~\ref{thm:quatleveltwo}, restated]
	For every constant deletion probability $d<1$, there exists an efficiently encodable code $\cC\subseteq\{A,C,G,T\}^{n+r}$ with redundancy $r=O(n/\log\log n)$ and balanced $GC$-content that can be efficiently reconstructed from $\exp(O(\log\log n)^{2/3})$ traces.
\end{thm}

Finally, two comments are in place regarding the choice of markers. First, the marker sequence $M=(AC)^\ell||(TG)^\ell$ may lead to hairpin formations when single stranded DNA is used. Hairpins are double-stranded folds, but may be easily controlled through addition of urea or through temperature increase. Second, repeats such as marker repeats are undesirable as they may lead to issues during DNA synthesis. To mitigate this issue, one can alternate marker sequences. For example, two valid marker options are $(AC)^\ell || (TG)^\ell$ and $(AG)^\ell || (TC)^\ell$, and any other marker where the sets of symbols used in each side are disjoint and $C$ and $G$ do not appear in the same side is appropriate for use in the construction. 

Note that alternating markers in turn requires alternating the inner codes used between markers. This can be accommodated in a straightforward manner. Suppose that the block $x^{(i)}$ precedes an $(AC)^\ell||(TG)^\ell$ marker. Then, we encode $x^{(i)}$ as usual with $\Enc'$ as defined in this section. However, if $x^{(i)}$ precedes an $(AG)^\ell||(TC)^\ell$ marker, then we encode $x^{(i)}$ by first computing $\Enc'(x^{(i)})$, and then swapping all $G$'s and $C$'s in the encoding. Observe that in both cases the encoding has balanced $GC$-content. Moreover, since $\cC'$ satisfies Property~\ref{prop:quat}, with high probability the trace of each block's encoding will not have long substrings containing only $A$'s and $C$'s (resp.\ $A$'s and $G$'s) before an $(AC)^\ell||(TG)^\ell$ marker (resp.\ $(AG)^\ell||(TC)^\ell$ marker). As before, this means that, with high probability, we can correctly split the full trace into the relevant sub-traces by alternately looking for long substrings composed of $A$'s and $C$'s only, and of $A$'s and $A$'s and $G$'s only. In fact, the end of such long substrings corresponds to the beginning of the traces of the $(TG)^\ell$ and $(TC)^\ell$ substrings of the marker, respectively.

\section{Reducing the number of traces for small constant deletion probability}\label{sec:reducetrace}

In Section~\ref{sec:marker}, we gave a construction of marker-based codes that require a few traces for reconstruction. A simple property of the inner code ensured that we can correctly identify all markers with high probability, effectively dividing the global trace into many independent, shorter traces. After this, we applied the state-of-the-art worst-case trace reconstruction algorithm from Lemma~\ref{lem:worstcase} on each short trace in order to obtain the desired codes.

It seems plausible, however, that one could design the inner code more carefully so that many fewer traces are needed to recover the short codewords contained between the markers. This is the main problem we address in this section. We design a code that, when used as the inner code in the construction from Section~\ref{sec:marker}, leads to an almost exponential reduction of the number of traces required for reconstruction with only a slight decrease in the code rate, provided that the deletion probability is a sufficiently small constant. The trace reconstruction algorithm we use is a variation of the algorithm for average-case trace reconstruction described in~\cite[Section 2.3]{HMPW08}.

Our starting point is a low redundancy code with the property that it can be reconstructed from $\poly(n)$ traces. We discuss this construction in Section~\ref{sec:lowred}. Then, in Section~\ref{sec:mainred} we show how to adapt this code so that it can be successfully used as an inner code in the marker-based construction introduced in Section~\ref{sec:marker}.

\subsection{Low redundancy codes reconstructable from polynomially many traces}\label{sec:lowred}

In what follows, we prove Theorem~\ref{thm:polytraces}. We restate the result for convenience.
\begin{thm}[Theorem~\ref{thm:polytraces}, restated]
	For small enough deletion probability $d$, there exists an efficiently encodable code $\cC\subseteq\bits^{n+r}$ with encoder $\Enc:\bits^n\to\bits^{n+r}$ and redundancy $r=O(\log n)$ that can be efficiently reconstructed from $\poly(n)$ traces with probability at least $1-\exp(-n)$.
\end{thm}
The code we construct to prove Theorem~\ref{thm:polytraces} will be the starting point for the proof of Theorem~\ref{thm:redtrace} in Section~\ref{sec:mainred}. Roughly speaking, our code encodes $n$-bit messages into codewords that are \emph{almost} $w$-subsequence-unique for $w=O(\log n)$, in the sense that all but the first $O(\log n)$ bits of the codeword comprise a $w$-subsequence-unique string. This is possible because an $\eps$-almost $k$-wise independent random variable over $\bits^n$ with the appropriate parameters is $w$-subsequence-unique with high probability. We make this statement rigorous in the following lemma. We note that the technique in the lemma below has already been used in~\cite{CJLW18} to obtain strings satisfying related properties, such as substring-uniqueness, with high probability.
\begin{lem}\label{lem:subsequnique}
	Let $g:\bits^t\to\bits^m$ be the function guaranteed by Lemma~\ref{lem:almostind} with $k=3w$ and $\eps=2^{-10w}$ for $w=100\log m$ (hence $t=O(\log m)$). Fix some $x\in\bits^m$ and define the random variable $Y=x+g(U_t)$. Then, with probability at least $1-1/\textnormal{poly}(m)$ it holds that $Y$ is $w$-subsequence-unique.
\end{lem}
\begin{proof}
First, note that $Y$ is $\eps$-almost $k$-wise independent. This proof follows along the same lines as the proof that a random string is $w$-subsequence-unique with high probability found in~\cite[Lemma 2.6]{HMPW08} with a few simple modifications.
	
Without loss of generality, fix $a$ and $b$ such that $a<b$, and fix distinct indices $i_1,\dots,i_w\in[b,\dots,b+1.1w)$. For convenience, let $\cS=\{i_1,\dots,i_w\}$, $\cS'=[b,b+1.1w)-\cS$, and $u=\min(a+w,b)$. Then,
	\begin{equation}\label{eq:totalprob}
		\Pr[Y_{\cS}=Y[a,a+w)]=\sum_{y,y'}\Pr[Y_{\cS}=Y[a,a+w),Y[a,u)=y,Y_{\cS'}=y'].
	\end{equation}
We now show that $Y[a,u)$ and $Y_{\cS'}$ completely determine $Y_\cS$ under the constraint $Y_\cS=Y[a,a+w)$. This can be seen by induction. First, we must have $Y_{i_1}=Y_a$, and $Y_a$ is determined by $Y[a,u)$ since $a< u$. Now, suppose that $Y_{i_1},\dots,Y_{i_j}$ 
are determined by $Y[a,u)$ and $Y_{\cS'}$. It must be the case that $Y_{i_{j+1}}=Y_{a+j}$. If $a+j< u$ or $a+j\in\cS'$, then $Y_{i_{j+1}}$ is determined by $Y[a,u)$ or $Y_{\cS'}$, respectively. On the other hand, if $a+j\geq u$ and $a+j\not\in\cS'$, then $Y_{a+j}=Y_{i_d}$ for some $d<j+1$. By the induction hypothesis, $Y_{i_d}$ is determined by $Y[a,u)$ and $Y_{\cS'}$, and hence $Y_{i_{j+1}}$ is, too.
	
As a result, we conclude that there exists a string $\overline{y}=(\overline{y}_1,\dots,\overline{y}_w)$ completely determined by 
$y$ and $y'$ such that
	\begin{equation}\label{eq:fixy}
		\Pr[Y_{\cS}=Y[a,a+w),Y[a,u)=y,Y_{\cS'}=y']=\Pr[Y_\cS=\overline{y},Y[a,u)=y,Y_{\cS'}=y'].
	\end{equation}
	Since $Y$ is $\eps$-almost $3w$-wise independent and fewer than $3w$ coordinates are fixed, we have
	\begin{equation}\label{eq:boundprob}
		\Pr[Y_\cS=\overline{y},Y[a,u)=y,Y_{\cS'}=y']\leq 2^{-1.1w-(u-a)}+2^{3w}\eps
	\end{equation}
	for all $y$ and $y'$. Combining~\eqref{eq:totalprob},~\eqref{eq:fixy}, and~\eqref{eq:boundprob}, we conclude that
	\begin{align*}
		\Pr[Y_{\cS}=Y[a,a+w)]\leq 2^{u-a}\cdot 2^{0.1w}(2^{-1.1w-(u-a)}+2^{3w}\eps)\leq 2^{-w}+2^{4.1w}\eps\leq 2^{-w+1},
	\end{align*}
	since $u-a\leq w$ and $\eps=2^{-10w}$. Since there are $\binom{1.1w}{w}$ choices for $\cS$ for each pair $(a,b)$ and fewer than $m^2$ possible pairs $(a,b)$, the probability that $Y$ is not $w$-subsequence-unique is at most
	\begin{align*}
	m^2 \binom{1.1w}{w}2^{-w+1}&=n^2\binom{1.1w}{w}2^{-w+1}\\
	&\leq m^2(11e)^{0.1w}2^{-w+1}\\
	&\leq 2m^2(1.415)^{-w}\\
	&\leq m^{-45},
	\end{align*}
	as desired.
\end{proof}
Lemma~\ref{lem:subsequnique} naturally leads to a simple, efficient candidate construction of the encoder $\Enc$: Given $x\in\{0,1\}^n$, we first iterate over all $z\in\{0,1\}^t$ until we find $z$ such that $x+g(z)$ is $w$-subsequence-unique. Most strings $z$ satisfy this, according to Lemma~\ref{lem:subsequnique}. Moreover, since $t=O(\log n)$, we can iterate over all such $z$ in time $\poly(n)$, and verify whether $x+g(z)$ is $w$-subsequence-unique for each $z$ in $\poly(n)$ time. To recover $x$ from $x+g(z)$ we need to provide $z$ to the receiver. Therefore, the encoder $\Enc$ for $\cC$ maps a message $x\in\bits^n$ to the codeword
\begin{equation}\label{eq:basiccode}
\Enc(x)=z||x+g(z)\in\bits^{n+t},
\end{equation}
where $z$ is the first string (in lexicographic order) such that $x+g(z)$ is $w$-subsequence-unique. Observe that the redundancy of $\cC$ is exactly $t=O(\log n)$.

\subsubsection{The trace reconstruction algorithm}\label{sec:modtracerec}

In this section, we describe an efficient trace reconstruction algorithm for $\cC$ that works whenever the deletion probability is a small enough constant, thus proving Theorem~\ref{thm:polytraces}. 
This algorithm works very similarly to the one introduced in~\cite{HMPW08} and described in Section~\ref{sec:tracerecuniquesubseq}. 
As before, we shall set $w=100\log n$, $v=w/d=O(\log n)$ and $j=(v-0.1w)(1-3d)=O(\log n)$. Given a codeword $c=\Enc(x)=z||x+g(z)$, we proceed as follows: 
First, we apply the algorithm from Lemma~\ref{lem:recoverfirst} to recover $z$ and the first $2v+w=O(\log n)$ bits of $y=x+g(z)$ with $\poly(n)$ traces (repeating the process $O(n)$ times if necessary) and success probability $1-\exp(-\Omega(n))$.
Now, suppose that we know $y_1,\dots,y_{i-1}$ for $i-1\geq 2v+w$. We show how to find $y_i$ with probability $1-\exp(-\Omega(n))$ from $\poly(n)$ traces, which concludes the proof of Theorem~\ref{thm:polytraces}. 

Let $T$ denote a trace of $c$. As in Section~\ref{sec:tracerecuniquesubseq}, we will look for a matching of $y[i-v-w,i-v)$ within $T$. However, we shall discard matchings that occur too early in $T$. More precisely, suppose that $y[i-v-w,i-v)$ is matched with $T[u-w,u)$. We call such a matching \emph{good} if $u-w>|z|$. If $T$ does not contain a good matching of $y[i-v-w,i-v)$, we discard it. Otherwise, if the first good matching occurs at $T[u-w,u)$, we let $V=T[u,\cdot]$ and discard the remaining bits of $T$. Our observations so far are summarized in the following lemmas.
\begin{lem}\label{lem:goodmatchprob}
	For $d$ small enough, the probability that a good matching occurs in $T$ is at least $2^{-(w+1)}$.
\end{lem}
\begin{proof}
	First, observe that the probability that no bit in $y[i-v-w,i-v)$ is deleted is exactly $(1-d)^w\geq 2^{-w}$. Given this, suppose that 
	$y[i-v-w,i-v)$ shows up in positions $T[U-w,U)$. Then, the probability that the given matching is good equals 
	$\Pr[U> |z|+w]$, and $|z|+w\leq Cw$ for a fixed constant $C>0$, since $|z|=O(\log n)$. Note that we may assume 
	$i-v-w\geq v=w/d$ since we have already learned the first $2v+w$ bits of $y$. We may also choose $d<1/10$ small enough 
	such that $v>Cw$. Then, we have
	\begin{equation*}
		\Pr[U\leq |z|+w]\leq \Pr[\mathsf{Bin}(2Cw,1-d)\leq Cw]<1/2,
	\end{equation*}
	where the last inequality follows from an application of the Chernoff bound. Concluding, the trace $T$ contains a 
	good matching with probability at least $1/2\cdot 2^{-w}=2^{-(w+1)}$.
\end{proof}
\begin{lem}\label{lem:goodmatchstruct}
	The probability that the last bit of a good matching in $T$ does not come from $y[i-v,i-v+0.1w)$ is at most $nd^{-w/100}\leq 2^{-100w}$ if $d$ is small enough. 
\end{lem}
\begin{proof}
	The probability that the event in question happens is at most the probability that more than $0.1w$ bits are deleted from some substring $y[b,b+1.1w)$. To see this, first note that the bits in a good matching must come from $y$. If at most $0.1w$ bits are deleted from every substring $y[b,b+1.1w)$, then the $w$ bits of the good matching in $T$ for $y[i-v-w,i-v)$ must be a subsequence of $y[b,b+1.1w)$ for some $b$, which means $y[i-v-w,i-v)$ appears as a subsequence of $y[b,b+1.1w)$. Since $y$ is $w$-subsequence-unique, for this to happen we must have $b\leq i-v-w$ and $b+1.1w\geq i-v$. Now suppose that the last bit of the good matching in $T$ does not come from $y[i-v,i-v+0.1w)$. Then, it must be the case that $y[i-v-w,i-v)$ is a subsequence of $y[b,i-v-1)$. Since $i-v-1< i-v$, this violates the $w$-subsequence-uniqueness propery of $y$.
	
For a fixed $b$, the probability that more than $0.1w$ bits are deleted from $y[b,b+1.1w)$ is at most $d^{-w/100}$ for $d$ small enough. The result then follows by a union bound, since there are fewer than $n$ choices for $b$.
\end{proof}

Let $E_{\textsf{good}}$ denote the event that a good matching occurs in $T$. From Lemma~\ref{lem:goodmatchprob} and the fact that we can efficiently check whether $\Egood$ occurred for $T$, it follows that we can efficiently estimate
\begin{equation*}
	\Pr[V_j=1|\Egood]
\end{equation*}
to within an error of, say, $2^{-100w}$ from $\poly(n)$ traces, with probability at least $1-\exp(-\Omega(n))$. Then, we proceed similarly to Section~\ref{sec:tracerecuniquesubseq}. Let $R$ be the random variable denoting the coordinate in $y$ of the last bit appearing in the good matching within $T$. We may then write
\begin{align*}
	\Pr[V_j=1|\Egood]&=\sum_{r=1}^n\Pr[R=r|\Egood]\Pr[V_j=1|R=r,\Egood]\\
	&=\eps_i(c)+\sum_{r=i-v}^{i-v+0.1w}\Pr[R=r|\Egood]\Pr[V_j=1|R=r]
\end{align*}
for $0\leq \eps_i(c)\leq 2^{-100w}$, by Lemma~\ref{lem:goodmatchstruct}. The second equality follows because, once $R=r$ is fixed, $V$ does not depend on whether $\Egood$ occurs or not, but only depends on the traces of $z$ and $y[1,r]$. Therefore, as in~\eqref{eq:UBoriginal} and~\eqref{eq:LBoriginal} we have
\begin{multline}\label{eq:UB}
	y_i=0\implies \Pr[V_j=1|\Egood]\leq \eps_i(c)+\sum_{r=i-v}^{i-v+0.1w}\Pr[R=r|\Egood]\sum_{\ell=r+1}^{i-1} P(\ell-r,j)s_\ell\\
	+\frac{1}{2}\sum_{r=i-v}^{i-v+0.1w}\Pr[R=r|\Egood]P(i-r,j)
\end{multline}
and
\begin{multline}\label{eq:LB}
y_i=1\implies \Pr[V_j=1|\Egood]\geq \eps_i(c)+\sum_{r=i-v}^{i-v+0.1w}\Pr[R=r|\Egood]\sum_{\ell=r+1}^{i-1} P(\ell-r,j)s_\ell\\
+\sum_{r=i-v}^{i-v+0.1w}\Pr[R=r|\Egood]P(i-r,j).
\end{multline}
Similarly to what was done in Section~\ref{sec:tracerecuniquesubseq}, since $i-r\leq v$ and $v=w/d$, the second part of Lemma~\ref{lem:basicpij} implies that $P(i-r,j)\geq 2^{-9w}$. Combining this result with Lemma~\ref{lem:goodmatchstruct} shows that the gap between the right hand sides of~\eqref{eq:UB} and~\eqref{eq:LB} is at least $2^{-(9w+1)}$. Each term $\Pr[R=r|\Egood]$ can be approximated to within an error of $2^{-90w}$ with probability at least $1-\exp(-\Omega(n))$ in time $\poly(n)$. This is accomplished by first using $z$ and the values $y_1,\dots,y_{i-1}$ that we have already recovered to estimate $\Pr[R=r|\Egood,R<i]$ to within a small enough error and with high probability. Then, the fact that $\Pr[R<i|\Egood]\geq 1-2^{-100w}$ and Lemma~\ref{lem:goodmatchstruct} imply that
\begin{equation*}
	|\Pr[R=r|\Egood]-\Pr[R=r|\Egood,R<i]|\leq 2\cdot 2^{-100w},
\end{equation*}
which in turn implies a good enough approximation for $\Pr[R=r|\Egood]$.

Since we know $y_1,\dots, y_{i-1}$, the discussion above suggests that we can approximate the right hand side of~\eqref{eq:UB} and~\eqref{eq:LB} to within an error of
\begin{equation*}
	2^{-100w}+n^2 \cdot 2^{-90w}\leq 2^{-80w}
\end{equation*}
with high probability. As already mentioned, we can estimate $\Pr[V_j=1|\Egood]$ to within error $2^{-100w}$ from $\poly(n)$ traces in time $\poly(n)$ with probability at least $1-\exp(-\Omega(n))$. Consequently, with probability $1-\exp(-\Omega(n))$ we can recover $y_i$ correctly from $\poly(n)$ traces, where the degree of this polynomial is independent of $i$. The success probability can be made at least $1-\exp(-Cn)$ for any fixed constant $C$ of our choice by repeating the process $O(n)$ times and taking the majority vote. Overall, we must recover fewer than $n$ positions of $y$, and each position requires $\poly(n)$ traces, where the degree of this polynomial is independent of the position of the bit. As a result, the total number of traces required is $\poly(n)$ and the overall success probability is $1-1/\poly(n)$. This proves Theorem~\ref{thm:polytraces}.

\subsection{Using the code within a marker-based construction}\label{sec:mainred}

Next, we combine the constructions from Sections~\ref{sec:marker} and~\ref{sec:lowred} with some additional modifications in order to prove Theorem~\ref{thm:redtrace}, which we restate here.
\begin{thm}[Theorem~\ref{thm:redtrace}, restated]
	For small enough deletion probability, there exists an efficiently encodable code $\cC$ with encoder $\Enc:\bits^n\to\bits^{n+r}$ and redundancy $r=O\left(\frac{n}{\log n}\right)$ that can be efficiently reconstructed from $\poly(\log n)$ traces with probability $1-1/\poly(n)$.
\end{thm}

The basic idea is that we would like to use the code designed in Section~\ref{sec:lowred} as the inner code $\cC'$ for the construction of $\cC$ in Section~\ref{sec:marker}. Then, we could apply the trace reconstruction algorithm from Section~\ref{sec:modtracerec} on each sub-trace and mitigate the use of worst-case trace reconstruction algorithms. This idea does not work as is, but some modifications to the code from Section~\ref{sec:lowred} will allow the construction to go through.

The first issue we have to address is for the inner code $\cC'$ to satisfy Property~\ref{prop:manyones}. If this property holds, then the reasoning of Section~\ref{sec:marker} implies that we can focus on the trace reconstruction problem for strings of the form
\begin{equation}\label{eq:specform}
	1^\ell || c || 0^\ell,
\end{equation}
where $c\in\cC'$ has length $O(\log^2 n)$ and $\ell=O(\log n)$, as long as we use a sub-polynomial number of traces in $n$. From here onwards we focus solely on this setting. If we were to directly apply the trace reconstruction algorithm from Section~\ref{sec:modtracerec}, we would run into a problem. For the aforementioned algorithm to work, we need to bootstrap it by recovering the first few bits of $c$ using the procedure described in Lemma~\ref{lem:recoverfirst}. However, in this case $c$ only appears after a run of length $\ell=O(\log n)$. Even though we know the previous bits, we still require $\poly(n)$ traces to recover the first bit of $c$ in this way, which is not acceptable as we want to use $\poly(\log n)$ traces. Consequently, we need an alternative bootstrapping method. Another issue we need to resolve is that the reconstruction algorithm from Section~\ref{sec:modtracerec} assumed that all but the first few bits of $c$ lead to a subsequence-unique string. However, this is not the case here, as we must deal with a string of the form $c||0^\ell$.

Before we proceed to describe a modified version of our code from Section~\ref{sec:lowred} that avoids the issues raised above, we first prove the following lemma.
\begin{lem}\label{lem:trailingsubsequnique}
	Let $g:\bits^t\to\bits^m$ be the function guaranteed by Lemma~\ref{lem:almostind} with $k=3w$ and $\eps=2^{-10w}$ for $w=100\log m$ (hence $t=O(\log m)$). For arbitrary $\ell$ and $x\in\bits^m$, define the random variable $Y=x+g(U_t)||0^\ell$. Then, with probability at least $1-1/\poly(m)$ we have that $Y$ satisfies the following property.
	\begin{prop}\label{prop:trailing}
		For any $a$ and $b$ such that $a+w\leq \min(m+1,b)$, we have that $Y[a,a+w)$ is not a subsequence of $Y[b,b+1.1w)$.
	\end{prop}
\end{lem}
\begin{proof}
	Fix a pair $(a,b)$ satisfying $a+w\leq \min(m+1,b)$ and let $\cS\subseteq [b,b+1.1w)$ be a set of size $w$. Let $u=\min(m+1,b+1.1w)$. Then, we have
	\begin{align*}
	\Pr[Y[a,a+w)=Y_\cS]=\sum_y \Pr[Y[a,a+w)=Y_\cS, Y[b,u)=y, Y[m+1,b+1.1w)=(0,\dots,0)].
	\end{align*}
	Observing that $Y_\cS$ is completely determined by $Y[b,u)$ and $Y[m+1,b+1.1w)$ and that $Y[m+1,b+1.1w)$ is fixed, we have
	\begin{align*}
	\Pr[Y[a,a+w)=Y_\cS]=\sum_y \Pr[Y[a,a+w)=y', Y[b,u)=y]
	\end{align*}
	for some $y'$ determined by $y$. Since $x+g(U_t)$ is $\eps$-almost $3w$-wise independent and fewer than $3w$ coordinates are fixed, we have
	\begin{equation*}
	\Pr[Y[a,a+w)=y', Y[b,u)=y]\leq 2^{-w-(u-b)}+2^{3w}\eps.
	\end{equation*}
	Therefore, it follows that
	\begin{equation*}
	\Pr[Y[a,a+w)=Y_\cS]\leq 2^{u-b}(2^{-w-(u-b)}+2^{3w}\eps)\leq 2^{-w}+2^{4.1w}\eps\leq 2^{-w+1}.
	\end{equation*}
	Since there are fewer than $m^3$ choices for pairs $(a,b)$ and $\binom{1.1w}{w}$ choices for $\cS$, from the union bound we conclude similarly to what we did in the proof of Lemma~\ref{lem:subsequnique} that the probability that the desired event does not happen is at most $m^3 \binom{1.1w}{w}2^{-w+1}\leq m^{-45}$.
\end{proof}
Intuitively, Lemma~\ref{lem:trailingsubsequnique} guarantees that $x+g(U_t)$ satisfies a stronger form of subsequence-uniqueness with high probability. In fact, not only is $x+g(U_t)$ $w$-subsequence-unique with high probability based on Lemma~\ref{lem:subsequnique}, but also is it impossible to find a substring of $x+g(U_t)$ that is a subsequence of $x+g(U_t)||0^\ell$ elsewhere.

We are now ready to describe our modified inner code $\cC'$ with encoder $\Enc':\bits^m\to \bits^{m+r'}$. On an input message $x\in\bits^{m}$, $\Enc'$ operates as follows:
\begin{enumerate}
	\item Set $x'=0^{\ell'}||x$ for $\ell'=10\ell=O(\sqrt{m})$. Let $m'=|x'|$ and set $w=100\log m'$;
	\item Iterate over all $z\in\bits^t$ for $t=O(\log m')=O(\log m)$ until a $z$ such that $x'+g(z)$ is $w$-subsequence-unique and simultaneously satisfies Properties~\ref{prop:manyonessmall} and~\ref{prop:trailing} is found. Such a string $z$ is guaranteed to exist because all such properties hold for $x'+g(U_t)$ with probability $1-o(1)$ (see Lemmas~\ref{lem:manyones},~\ref{lem:subsequnique}, and~\ref{lem:trailingsubsequnique}). Moreover, whether $x'+g(z)$ satisfies all three properties can be checked in time $\poly(m)$;
	
	\item Obtain $z'$ from $z$ by setting $z'=\Enc_{\mathsf{edit}}(0||z)$, where $\Enc_{\mathsf{edit}}$ is the encoder of the systematic code $\cC_{\mathsf{edit}}$ from Lemma~\ref{lem:systcode} robust against $|z|/2$ edit errors and with redundancy $O(|z|)=O(\log m)$. Here, $d$ is assumed to be a small enough constant so that $5d|z'|<|z|/2$, i.e., $\cC_{\mathsf{edit}}$ can correct a $5d$-fraction of edit errors in $z'$. This is possible because $|z'|=O(|z|)$;
	\item Define $\Enc'(x)=z'||x'+g(z)=z'||y'$.
\end{enumerate}
For a given message $x\in\bits^m$, we can compute $\Enc'(x)$ in time $\poly(m)$. Furthermore, recalling that $m=\log^2 n$ in the construction of Section~\ref{sec:basicconst}, the redundancy of $\cC'$ is
\begin{equation*}
	r'=|z'|+\ell'=O(\log m+\sqrt{m})=O(\sqrt{m})=O(\log n).
\end{equation*}
If we use $\cC'$ as the inner code in the construction of $\cC$ from Section~\ref{sec:basicconst}, then according to~\eqref{eq:redmarker} we obtain an overall redundancy $r=O\left(\frac{n}{\log n}\right)$ for $\cC$, as desired. It is also easy to see that $\cC'$ satisfies Property~\ref{prop:manyones}. By the choice of $z$, we have $w(y'[a,a+w))\geq 0.4w$ for every $a$ and $w=100\log m'$. Therefore, for any substring $s$ such that $|s|=\sqrt{m}$ we have
\begin{equation*}
w(s)\geq 0.4|s|-|z'|\geq 0.39|s|
\end{equation*}
provided $m$ is large enough, since $|z'|=O(\log m)$. As a result, the reasoning used in Section~\ref{sec:basicconst} applies to this choice of $\cC'$. To prove Theorem~\ref{thm:redtrace}, it remains to give a trace reconstruction algorithm to recover strings of the form $1^\ell ||\Enc'(x)||0^\ell$ from $\poly(m)=\poly(\log n)$ traces with probability, say, $1-n^{-10}$.

To address the problem, suppose we already have such an algorithm, and call it $\mathcal{A}$. Recall~\eqref{eq:genboundfail} and the definition of the event $E^{(i)}_\mathsf{indFail}$ from Section~\ref{sec:basicconst}. Instantiating $E^{(i)}_\mathsf{indFail}$ with algorithm $\mathcal{A}$ leads to the bound $\Pr[E^{(i)}_\mathsf{indFail}]\leq n^{-10},$ for all $i$. Combining this observation with~\eqref{eq:genboundfail} allows us to conclude that the probability that we successfully recover $c\in\cC$ from $\poly(\log n)$ i.i.d.\ traces of $c$ is at least $1-2/n$. Similarly to Section~\ref{sec:basicconst}, we can boost the success probability to $1-1/p(n)$ for any fixed polynomial of our choice by repeating the process $O(\log n)$ times and by taking a majority vote.

\subsubsection{The trace reconstruction algorithm}

Next, we analyze an algorithm for recovering strings of the form $1^\ell ||\Enc'(x)||0^\ell$ from $\poly(m)=\poly(\log n)$ traces with probability $1-1/\poly(n)$. As discussed before, we proceed by adapting the algorithm from Section~\ref{sec:modtracerec}, which in turn is a modified version of the algorithm from~\cite{HMPW08} described in Section~\ref{sec:tracerecuniquesubseq}.

The main difference between the current and the two previously discussed settings is that the original bootstrapping technique cannot be applied, 
as $\Enc'(x)$ is enclosed by two long runs. We start by showing that the structure of $\Enc'$ allows for a simple alternative bootstrapping method.

Recall that $c=\Enc'(x)=z'||y'$, where $y'=x'+g(z)$ and the first $O(\sqrt{m})$ bits of $x'$ are zero. Therefore, if we can recover $z$ from a few traces of $1^\ell ||c||0^\ell$, then we can recover the first $O(\sqrt{m})$ bits of $y'$, which suffices for bootstrapping, by simply computing $g(z)$. The following lemma states that we can recover $z$ with high probability from $O(\log n)$ traces.
\begin{lem}\label{lem:recoverz}
	There is an algorithm that recovers $z$ from $O(\log n)$ traces of $1^\ell ||c||0^\ell$ with probability at least $1-n^{-10}$.
\end{lem}
\begin{proof}
	We begin by recalling that $z'=\Enc_{\mathsf{edit}}(0||z)$, and that $\cC_{\mathsf{edit}}$ is systematic. This means $z'_1=0$, and so with probability $1-d$ the first $0$ appearing in the trace will correspond to $z'_1$. 
	
	Given a trace $T$ of $1^\ell ||c||0^\ell$, we proceed as follows: Let $u$ denote the position of the first $0$ in $T$. Then, we take $\tilde{z}=T[u,u+(1-d)|z'|)$, feed $\tilde{z}$ into $\Dec_{\mathsf{edit}}$, and let the corresponding output be our guess for $z$. The probability that this procedure fails to yield $z$ is at most the probability that $z'_1$ was deleted, plus the probability that $\tilde{z}$ is too far away in edit distance from $z'$ given that $z'_1$ was not deleted. We proceed to bound both probabilities. First, the probability that $z'_1$ is deleted is exactly $d$. Second, we assume $z'_1$ is not deleted and let $L$ denote the length of the trace of $z'[2,\cdot]$ within $T$. We have $\E[L]=(1-d)(|z'|-1)$. Therefore, a Chernoff bound gives
	\begin{equation*}
		\Pr[L\geq (1-3d)(|z'|-1)]\leq \exp\left(-\frac{2d^2}{1-d}(|z'|-1)\right).
	\end{equation*}
	Since $d$ is a constant and $|z'|=\Theta(\log m)$, we conclude that for $m$ large enough we have
	\begin{equation*}
		\Pr[|L-(1-d)(|z'|-1)|\geq 2d(|z'|-1)]<1/5.
	\end{equation*}
	 As a result, with probability at least $4/5$ we have that $\tilde{z}$ is within edit distance $5d|z'|<|z|/2$ from $z'$. If this distance condition holds, then $\Dec_{\mathsf{edit}}(\tilde{z})=z$.
	
	In summary, the procedure fails to return $z$ with probability at most $d+1/5<1/4$ if $d$ is small enough. Repeating this procedure $O(\log n)$ times and taking the majority vote ensures via a Chernoff bound that we can recover $z$ from $O(\log n)$ traces with success probability at least $1-1/p(n),$ for $p$ any choice of a fixed polynomial.
\end{proof}

Once $z$ has been recovered, the bits of $1^\ell ||c||0^\ell=1^\ell ||z'||y'||0^\ell$ are known up to and including the first $\ell'=O(\sqrt{m})$ bits of $y'$. Our last task is to recover the remaining bits of $y'$, and given that we have sufficiently many initial bits from $y'$ we may to this end use the ideas from Section~\ref{sec:modtracerec}. The differences with respect to Section~\ref{sec:modtracerec} are the following:
\begin{itemize}
	\item Instead of $y$, we use $y''=y'||0^\ell$;
	\item We are only interested in recovering $y''_i$ for $\ell'<i\leq |y'|$, since we already know all other bits of $y''$;
	\item We change the threshold used to declare that a matching is good: In this case, if $T$ is a trace of $1^\ell ||c||0^\ell$ and $y''[i-v-w,i-v]$ is matched with $T[u-w,u]$, then the matching is said to be \emph{good} if $u-w>\ell+|z'|$. This change ensures that the bits in a good matching always come from $y''=y'||0^\ell$.
\end{itemize}
Two key lemmas now follow from the previous discussion. Their statements and proofs are similar to the ones of Lemmas~\ref{lem:goodmatchprob} and~\ref{lem:goodmatchstruct} from Section~\ref{sec:modtracerec}, respectively, and we hence only discuss relevant differences. 
Henceforth, we use $T$ to denote a trace of $1^\ell ||c||0^\ell$.
\begin{lem}\label{lem:goodmatchprob2}
	The probability that a good matching occurs in $T$ is at least $2^{-(w+1)}$.
\end{lem}

\begin{lem}\label{lem:goodmatchstruct2}
	For $\ell'<i\leq |y'|$, the probability that the last bit of a good matching in $T$ does not come from $y''[i-v,i-v+0.1w]$ is at most $nd^{-w/100}\leq 2^{-100w}$ if $d$ is small enough. 
\end{lem}
\begin{proof}
	Similarly to the proof of Lemma~\ref{lem:goodmatchstruct}, the probability of the event in the statement of the lemma is upper bounded by the probability that more than $0.1w$ bits are deleted from some substring $y''[b,b+1.1w)$. We explain next why this is true. First, note that the bits in a good matching must come from $y''$. Suppose that at most $0.1w$ bits are deleted from every substring $y''[b,b+1.1w)$. Then, $y''[i-v-w,i-v)$ must be a subsequence of $y''[b,b+1.1w)$ for some $1\leq b\leq |y''|-1.1w$. We distinguish two cases:
	\begin{itemize}
		\item $b+1.1w>|y'|$:
		
		Recalling that $v=w/d$, we have $i-v\leq |y'|-w/d\leq |y'|-1.1w\leq \min(|y'|+1,b)$, and so Property~\ref{prop:trailing} holds for $y''[i-v-w,i-v)$. Therefore, $y''[i-v-w,i-v)$ cannot be a subsequence of $y''[b,b+1.1w)$ for any $b$ such that $b+1.1w>|y'|+1$. Consequently, we only need to consider values of $b$ such that $b+1.1w\leq |y'|$;
		
		\item $b+1.1w\leq |y'|$: 
		
		Since $y'$ is $w$-subsequence-unique, we must have $b\leq i-v-w$ and $b+1.1w\geq i-v$. This implies the desired result as in the proof of Lemma~\ref{lem:goodmatchstruct};
	\end{itemize}
	The remainder of the proof follows along the lines of the proof of Lemma~\ref{lem:goodmatchstruct}.
\end{proof}

Lemmas~\ref{lem:goodmatchprob2} and~\ref{lem:goodmatchstruct2} imply that we can recover $y''_i$ with probability $1-1/\poly(n)$ via the same reasoning of Section~\ref{sec:modtracerec} with the small differences described above. The number of traces required to recover $y''_i$ is polynomial in the length of $1^\ell ||c||0^\ell$, which equals
\begin{equation*}
	2\ell + |z'|+|y'|=O(\sqrt{m}+\log m+m)=O(m).
\end{equation*}
Since $m=\log^2 n$, it follows that we can recover $y''_i$ with probability $1-1/\poly(n)$ from $\poly(\log n)$ traces. In particular, the success probability can be assumed to be at least $1-1/p(n)$ for a fixed polynomial of our choice since we can repeat the process $O(\log n)$ times and take the majority vote while still requiring $\poly(\log n)$ traces. Since Lemma~\ref{lem:recoverz} asserts that $O(\log n)$ traces suffice to recover $z$ with high probability, and we need to recover $m=\log^2 n$ bits of $y''$, we overall require $\poly(\log n)$ traces to recover $1^\ell ||c||0^\ell$ with probability $1-1/\poly(n)$. This concludes the proof of Theorem~\ref{thm:redtrace}.

\section{Open problems}\label{sec:open}

The newly introduced topic of coded trace reconstruction can be extended in many different directions, in which the following problems are of interest:
\begin{itemize}
	\item Coded trace reconstruction based on a more general set of edit errors. Our constructions only work against i.i.d.\ deletions. However, errors imposed by the nanopore sequencing process in DNA-based storage also include insertions and substitutions, and the errors may be symbol and context-dependent;
	
	\item Narrowing the gap between upper and lower bounds on the number of traces used in reconstruction. The current gap between lower and upper bounds for trace reconstruction is almost exponential in the string length, and there is no study of lower bounds on coded trace reconstruction as a function of the redundancy of the code. We suspect that our current construction of efficient codes are suboptimal;
	
	\item The main contributors to the redundancy in all but one of our constructions are the markers placed between blocks. It is of interest to find other constructions that either do not use markers, or that use shorter markers that can still be integrated into the trace reconstruction process.
\end{itemize}

\bibliographystyle{IEEEtran}
\bibliography{trace-refs}

\appendix

\section{Tradeoff between redundancy and relative Hamming distance for efficiently encodable/decodable binary codes}\label{app:hamming}

We establish next the existence of efficiently encodable and decodable codes $\cC_{\mathsf{Ham}}\subseteq\bits^n$ with encoder $\Enc_{\mathsf{Ham}}:\bits^{n_0}\to\bits^n$, relative Hamming distance at least $30/\log^2 n_0$, and redundancy $n-n_0=O\left(n_0\frac{\log\log n_0}{\log n_0}\right)$. We make use of the Zyablov bound~\cite[Sections 10.2]{GRS18}, which states that an efficiently encodable and decodable binary code with rate $R$ and relative Hamming distance $\delta$ exists for $R$ and $\delta$ satisfying
\begin{equation*}
	R\geq \max_{0<r<1-h(\delta+\eps)} r\left(1-\frac{\delta}{h^{-1}(1-r)-\eps}\right),
\end{equation*}
where $\eps>0$ is arbitrary. This bound is achieved by concatenating a Reed-Solomon outer code with an inner linear code lying on the Gilbert-Varshamov bound. The relevant generator matrices can be constructed in time polynomial in the block length of the concatenated code, and one can efficiently correct substitution errors up to half of its designed distance via generalized minimum distance decoding~\cite[Section 11.3]{GRS18}.

We set $\delta=30/\log^2 n_0$ and $\eps=1/\log^2 n_0$. Then,
\begin{equation*}
	h(\delta+\eps)=h\left(31/\log^2 n_0\right)\leq 62\cdot \frac{\log\log n_0}{\log^2 n_0}
\end{equation*}
for $n_0$ large enough. The inequality follows from the fact that $h(p)\leq -2p\log p$ for $p$ small enough. We set $r=1-\frac{\log\log n_0}{\log n_0}$ and observe that $r<1-h(\delta+\eps)$ for $n_0$ large enough. Moreover, we have
\begin{equation*}
	h^{-1}(1-r)=h^{-1}\left(\frac{\log\log n_0}{\log n_0}\right)\geq \frac{1}{2\log n_0}.
\end{equation*}
As a result, 
\begin{align}\label{eq:zyablov1}
	1-\frac{\delta}{h^{-1}(1-r)-\eps}&\geq 1-\frac{\delta}{1/2\log n_0-1/\log^2 n_0}\nonumber\\
	&=1-\frac{10/\log^2 n_0}{1/2\log n_0-1/\log^2 n_0}\nonumber\\
	&\geq 1-\frac{40}{\log n_0}.
\end{align}
Combining~\eqref{eq:zyablov1} with the previously described choice for $r$, it follows that
\begin{equation*}
	R\geq r\left(1-\frac{40}{\log n_0}\right)>1-\frac{2\log\log n_0}{\log n_0},
\end{equation*}
and hence the corresponding redundancy is $O\left(n_0\frac{\log\log n_0}{\log n_0}\right)$, as desired.
\end{document}